\documentclass[11pt,reqno,tbtags]{amsart}
\usepackage{amssymb}
\usepackage{natbib}
\bibpunct[, ]{[}{]}{;}{n}{,}{,}

\numberwithin{equation}{section}

\allowdisplaybreaks

\newtheorem{theorem}{Theorem}[section]
\newtheorem{lemma}[theorem]{Lemma}

\newtheorem{corollary}[theorem]{Corollary}

\theoremstyle{definition}

\newtheorem{remark}[theorem]{Remark}

\newtheorem*{ack}{Acknowledgement}

\theoremstyle{remark}

\newenvironment{romenumerate}[1][0pt]{
\addtolength{\leftmargini}{#1}\begin{enumerate}
 }{\end{enumerate}}

\newcounter{oldenumi}
{\setcounter{oldenumi}{\value{enumi}}
\begin{romenumerate} \setcounter{enumi}{\value{oldenumi}}}
{\end{romenumerate}}

\newcounter{thmenumerate}
\newenvironment{thmenumerate}
{\setcounter{thmenumerate}{0}%
 \def\item{\par
 \refstepcounter{thmenumerate}\textup{(\roman{thmenumerate})\enspace}}
}
{}

\newcounter{romxenumerate}   

\newcounter{xenumerate}   

\newcommand{\refT}[1]{Theorem~\ref{#1}}

\newcommand{\refL}[1]{Lemma~\ref{#1}}
\newcommand{\refR}[1]{Remark~\ref{#1}}
\newcommand{\refS}[1]{Section~\ref{#1}}

\newcommand{\refApp}[1]{Appendix~\ref{#1}}

\newcommand{\refand}[2]{\ref{#1} and~\ref{#2}}

\newcommand\marginal[1]{\marginpar{\raggedright\parindent=0pt\tiny #1}}

\begingroup
  \divide\count255 by 60
  \multiply\count255 by -60
  \advance\count255 by \time
  \ifnum \count255 < 10 \xdef\klockan{\the\count1.0\the\count255}
  \else\xdef\klockan{\the\count1.\the\count255}\fi
\endgroup

\newcommand{\sumj}{\sum_{j=0}^\infty}
\newcommand{\sumk}{\sum_{k=0}^\infty}

\newcommand{\sumn}{\sum_{n=0}^\infty}
\newcommand{\sumni}{\sum_{n=1}^\infty}
\newcommand{\sumaa}{\sum_{\aaa\in\cAx}}

\newcommand{\sumik}{\sum_{i=1}^k}
\newcommand{\sumin}{\sum_{i=1}^n}
\newcommand{\prodin}{\prod_{i=1}^n}

\newcommand{\sumkoooo}{\sum_{k=-\infty}^\infty}
\newcommand{\summoooo}{\sum_{m=-\infty}^\infty}

\newcommand\set[1]{\ensuremath{\{#1\}}}

\newcommand\bigpar[1]{\bigl(#1\bigr)}
\newcommand\Bigpar[1]{\Bigl(#1\Bigr)}

\newcommand\lrpar[1]{\left(#1\right)}

\newcommand\xcpar[1]{\{#1\}}

\newcommand\Bigcpar[1]{\Bigl\{#1\Bigr\}}

\newcommand\abs[1]{|#1|}

\def\rompar(#1){\textup(#1\textup)}    
\newcommand\xfrac[2]{#1/#2}

\newcommand\Bigparfrac[2]{\Bigpar{\frac{#1}{#2}}}

\def\xexp(#1){e^{#1}}
\newcommand\ceil[1]{\lceil#1\rceil}
\newcommand\floor[1]{\lfloor#1\rfloor}

\newcommand\ntoo{\ensuremath{{n\to\infty}}}

\newcommand\ttoo{\ensuremath{{t\to\infty}}}
\newcommand\bmin{\wedge}
\newcommand\bmax{\vee}

\newcommand\punkt[1]{\if.#1\else.\spacefactor1000\fi{#1}}
\newcommand\iid{i.i.d\punkt}    
\newcommand\ie{i.e\punkt}
\newcommand\eg{e.g\punkt}

\newcommand{\aex}{a.e\punkt}
\newcommand\whp{w.h.p\punkt}

\newcommand\ii{\mathrm{i}}

\newcommand{\tend}{\longrightarrow}
\newcommand\dto{\overset{\mathrm{d}}{\tend}}
\newcommand\pto{\overset{\mathrm{p}}{\tend}}
\newcommand\asto{\overset{\mathrm{a.s.}}{\tend}}
\newcommand\eqd{\overset{\mathrm{d}}{=}}

\newcommand\bbR{\mathbb R}

\newcommand\bbZ{\mathbb Z}

\newcounter{CC}
\newcounter{cc}
\newcounter{ppsi}
\newcommand{\ppsi}{\stepcounter{ppsi}\ppsix} 
\newcommand{\ppsix}{\psi_{\arabic{ppsi}}}     
\newcommand{\ppsidef}[1]{\xdef#1{\ppsix}}     

\newcommand\E{\operatorname{\mathbb E{}}}
\renewcommand\P{\operatorname{\mathbb P{}}}
\newcommand\Var{\operatorname{Var}}

\newcommand\Exp{\operatorname{Exp}}
\newcommand\Po{\operatorname{Po}}
\newcommand\Bi{\operatorname{Bi}}

\newcommand\Be{\operatorname{Be}}

\newcommand\AsN{\operatorname{AsN}}

\newcommand\ga{\alpha}
\newcommand\gb{\beta}
\newcommand\gd{\delta}
\newcommand\gD{\Delta}
\newcommand\gf{\varphi}

\newcommand\gl{\lambda}

\newcommand\gs{\sigma}
\newcommand\gss{\sigma^2}
\newcommand\eps{\varepsilon}

\newcommand\cA{\mathcal A}

\newcommand\cD{\mathcal D}

\newcommand\cT{{\mathcal T}}

\newcommand\ett[1]{\boldsymbol1[#1]} 

\def\[#1]{[\![#1]\!]}

\newcommand\qq{^{1/2}}

\newcommand\qqw{^{-1/2}}
\newcommand\qqqw{^{-1/3}}

\newcommand\qw{^{-1}}
\newcommand\qww{^{-2}}
\newcommand\qwww{^{-3}}
\newcommand\qqqb{^{2/3}}

\renewcommand{\=}{:=}

\newcommand\intoo{\int_0^\infty}
\newcommand\intoooo{\int_{-\infty}^\infty}
\newcommand\oi{[0,1]}
\newcommand\ooo{[0,\infty)}
\newcommand\oooo{(-\infty,\infty)}
\newcommand\setoi{\set{0,1}}

\newcommand\dtv{d_{\mathrm{TV}}}

\newcommand\dd{\,\textup{d}}
\newcommand\ddx{\textup{d}}

\newcommand\rhs{right-hand side}

\newcommand\sss[1]{^{(#1)}}
\newcommand\sssn{\sss n}
\newcommand\hS{\widehat S}

\newcommand\hg{\widehat g}
\newcommand\hnu{\widehat\nu}

\newcommand\hpsi{\widehat\psi}
\newcommand\ctoo{\cT_\infty}
\newcommand\xxo{X_0^*}
\newcommand\aaa{{\boldsymbol{\alpha}}}
\newcommand\bbb{\boldsymbol{\beta}}
\newcommand\cAx{\cA^*}
\newcommand\xY{W}
\newcommand\tY{\widetilde\xY}
\newcommand\tZ{\widetilde Z}
\newcommand\zj{Z_{j}}
\newcommand\zjn{Z_{jn}}
\newcommand\tzj{\tZ_{j\gl}}
\newcommand\ggx{f_1}
\newcommand\ggh{g}
\newcommand\hggh{\widehat\ggh}
\newcommand\FF{f_0} 
\newcommand\frakt[1]{\xcpar{#1}}
\newcommand\Bigfrakt[1]{\Bigcpar{#1}}
\newcommand\kk{\kappa}
\newcommand\dl{L}
\newcommand\dtm{\dl^{\mathsf T}_M}
\newcommand\dtx[1]{\dl^{\mathsf T}_{#1}}
\newcommand\dkr{\dl^{\mathsf K}_R}
\newcommand\dkv{D_{K,V}}
\newcommand\hgs{\widetilde\sigma}

\newcommand\tgs{\widehat\sigma}
\newcommand\hgss{\hgs^2}
\newcommand\tD{\widetilde D}
\newcommand\tK{\widetilde K}
\newcommand\tKi{\widetilde K_1}
\newcommand\tnu{\widetilde \nu}
\newcommand\sqrtvv{\sqrt{V_2}}
\newcommand\gcdpq{\gcd(\ln p,\ln q)}
\newcommand\pat{^{P}}
\newcommand\psibj{\psi_{bj}}
\newcommand\ppsiq[1]{\psi_{\arabic{ppsi}#1}}
\newcommand\gdm{\gD_M}

\newcommand\CS{Cauchy--Schwarz}
\newcommand\CSineq{\CS{} inequality}

\newcommand\REM[1]{{\raggedright\texttt{[#1]}\par\marginal{XXX}}}

\newenvironment{comment}{\setbox0=\vbox\bgroup}{\egroup} 

\newcommand\citetq[2]{\citeauthor{#2} \cite[{\frenchspacing #1}]{#2}}

\newcommand\urladdrx[1]{{\urladdr{\def~{{\tiny$\sim$}}#1}}}

\begin{document}
\title
{Renewal theory in analysis of tries and strings}

\date{December 11, 2009} 

\author{Svante Janson}
\address{Department of Mathematics, Uppsala University, PO Box 480,
SE-751~06 Uppsala, Sweden}
\email{svante.janson@math.uu.se}
\urladdrx{http://www.math.uu.se/~svante/}

\dedicatory{To my colleague and friend Allan Gut on the occasion of
  his retirement}


\begin{comment}  

60K Special processes
60K05 Renewal theory
60K10 Applications (reliability, demand theory, etc.)
60K15 Markov renewal processes, semi-Markov processes
60K20 Applications of Markov renewal processes (reliability, queueing
      networks, etc. ) 
60K35 Interacting random processes; statistical mechanics type models;
      percolation theory [See also 82B43, 82C43] 

68 Computer science
68P Theory of data
68P05 Data structures
68P10 Searching and sorting

\end{comment}

\begin{abstract} 
We give a survey of a number of simple applications of renewal theory
to problems on random strings and tries: insertion depth, size,
insertion mode and imbalance of tries; variations for $b$-tries and
Patricia tries; Khodak and Tunstall codes.
\end{abstract}

\maketitle

\section{Introduction}\label{S:intro}

Although it long has been realized that renewal theory is a useful
tool in the study of random strings and related structures, it has not
always been used to its full potential.
The purpose of the present paper is to give a survey presenting in a
unified way some 
simple applications of renewal theory to a number of problems
involving random strings, in particular several problems on tries,
which are tree structures constructed from strings.
(Other applications of renewal theory to problems on random trees are
given in, \eg, \cite{DGrubel} and \cite{SJ165}.)

Since our purpose is to illustrate a method rather than to prove new
results, we present a number of problems in a simple form without
trying to be as general as possible. In particular, for simplicity we
exclusively consider random strings in the alphabet \setoi, and assume
that the ``letters'' (bits) $\xi_i$ in the strings are \iid.
Note, however, that the methods below are much more widely applicable and
extend in a straightforward way to larger alphabets. The methods also,
at least in principle,  extend to, for
example, Markov sources where $\xi_i$ is a Markov chain. 
(See \eg{} \citetq{Section  2.1}{Szp-book} 
and \citet{CFV01} for various interesting
probability models of random strings.
Renewal theory for Markov chains is treated for example by
\citet{Kesten} and \citet{AN78}.)
Indeed, one of the purposes of this paper is to make propaganda for
the use of renewal theory to study \eg{} Markov models, even if we do
not do this in the present paper. (Some such results may appear elsewhere.)

The results below are (mostly) not new;
they have earlier been proved by other methods, in particular Mellin
transforms.
(We try to give proper references for the theorems, but we do not
attempt to cover the large literature on random tries and strings in
any completeness.)
Indeed, such methods often provide sharper results, with better error
bounds or higher order terms, and these methods too certainly are
important. Nevertheless, we believe that renewal
theory often is a valuable method that yields the leading terms in a
simple and intuitive way, and that it ought to be more widely used for
this type of problems. Moreover, as said above, this method may be
easier to extend to other situations. (Further, it gives one
explanation for the oscillatory terms that often appear, as an
instance of the arithmetic case in renewal theory. Note that oscillatory terms
become much less common for larger alphabets, except when all letters
are equiprobable, because it is more difficult to be arithmetic, see
\refApp{App}.)

We treat a number of problems on random tries in Sections
\ref{Strie}--\ref{Sexp} and \ref{Sins} (insertion depth, imbalance,
size, insertion mode).
We consider 
$b$-tries in \refS{Sb} and Patricia tries in  \refS{Spat}. 
Tunstall and Khodak codes are studied in \refS{Scodes}. 
A random walk in a region bounded by two crossing lines is studied in
\refS{SX}. 
The standard
results from renewal theory that we use are for convenience collected
in \refApp{App}.

\subsection*{Notation}

We use $\pto$ and $\dto$ for convergence in probability and
in distribution, respectively.

If $Z_n$ is a sequence of random variables and 
$\mu_n$ and $\gss_n$ are sequences of real numbers with $\gss_n>0$
(for large $n$, at least), then
$Z_n\sim\AsN(\mu_n,\gss_n)$ means that $(Z_n-\mu_n)/\gs_n\dto N(0,1)$.


We denote the fractional part of a real number $x$ by 
$\frakt x\=x-\floor x$.

\begin{ack}
  I thank 
Allan Gut and Wojciech Szpankowski for 
inspiration and helpful  discussions.
\end{ack}

\section{Preliminaries}\label{Sprel}

Suppose that $\Xi\sss1,\Xi\sss2,\dots$ is an \iid{} sequence of random
infinite strings $\Xi\sssn=\xi_1\sssn\xi_2\sssn\dotsm$, with letters
$\xi_i\sssn$ in an alphabet $\cA$. (When
the superscript $n$ does not matter we drop it; we thus write
$\Xi=\xi_1\xi_2\dotsm$ for a generic string in the sequence.)
For simplicity, we consider only the case $\cA=\setoi$, and further
assume that the individual letters $\xi_i$ are \iid{} with
$\xi_i\sim\Be(p)$ for some fixed $p\in(0,1)$, \ie, 
$\P(\xi_i=1)=p$ and $\P(\xi_i=0)=q\=1-p$.

Given a finite string $\ga_1\dotsm\ga_n\in\cA^n$, let
$P(\ga_1\dotsm\ga_n)$ be the probability that the random string $\Xi$
begins with $\ga_1\dotsm\ga_n$. 
In particular,
for a single letter, $P(0)=q$ and $P(1)=p$, and in general
\begin{equation}\label{paaa}
  P(\ga_1\dotsm\ga_n)=\prodin P(\ga_i)=\prodin p^{\ga_i}q^{1-\ga_i}.
\end{equation}

Given a random string $\xi_1\xi_2\dotsm$, we define
\begin{equation}\label{x}
  X_i\=-\ln P(\xi_i)
=-\ln\bigpar{ p^{\xi_i}q^{1-\xi_i}}
=
\begin{cases}
  -\ln q, & \xi_i=0, \\
  -\ln p, & \xi_i=1.
\end{cases}
\end{equation}
Note that $X_1,X_2,\dots$ is an \iid{} sequence of positive random
variables with
\begin{align}
  \E X_i &= H:=-p\ln p -q\ln q, \label{H}
\intertext{the usual \emph{entropy} of each letter $\xi_i$, and}
  \E X^2_i& = H_2:=p\ln^2 p +q\ln^2q,
\\
\Var X_i &= H_2-H^2=pq (\ln p-\ln q)^2=pq \ln^2 (p/q).
\end{align}
Note that the case $p=q=1/2$ is special; in this case $X_i=\ln 2$ is
deterministic and $\Var X_i=0$; for all other $p\in(0,1)$, $0<\Var X_i<\infty$.

By \eqref{x}, $X_i$ is supported on \set{\ln(1/p),\ln(1/q)}. 
It is well-known, both in renewal theory and in the analysis of tries,
that one frequently has to distinguish between two cases: the
\emph{arithmetic} (or \emph{lattice})
case  when the support is a subset of $d\bbZ$ for
some $d>0$, and 
the \emph{non-arithmetic} (or \emph{non-lattice})
case when it is not, see further \refApp{App}.
For $X_i$ given by \eqref{x}, this yields the following cases:
\begin{description}
  \item[arithmetic]
The ratio $\ln p/\ln q$ is rational. 
More precisely, 
$X_i$ then is $d$-arithmetic, where $d$ equals $\gcd(\ln p,\ln q)$, the
largest positive real 
number such that $\ln p$ and $\ln q$ both are integer multiples of $d$.
If $\ln p/\ln q=a/b$, where $a$ and $b$ are
relatively prime positive integers, then 
\begin{equation}\label{d}
d=
\gcd(\ln p,\ln q)
=\frac{|\ln p|}a=\frac{|\ln q|}b.  
\end{equation}

  \item[non-arithmetic]  
The ratio $\ln p/\ln q$ is irrational. 
\end{description}

We let $S_n$ denote the partial sums of $X_i$: $S_n\=\sumin X_i$.
Thus
\begin{equation}\label{sop}
  P(\xi_1\dotsm\xi_n)=\prodin P(\xi_i)
=\prodin e^{-X_i}
=e^{-S_n}.
\end{equation}
(This is a random variable, since it depends on the random string
$\xi_1\dotsm\xi_n$; it can be interpreted as the probability that
another random string $\Xi\sss j$ begins with the same $n$ letters as
observed.)

We introduce the standard renewal theory notations 
(see \eg{} \citetq{Chapter  2}{GutSRW}), for $t\ge0$ and $n\ge1$,
\begin{align}
  \nu(t)&\=\min\set{n:S_n> t}   
, \label{nu}
\\
  F_n(t)&\=\P(S_n\le t)  
=\P (\nu(t)> n),
\\
 U(t)&\=\E \nu(t)=\sumn F_n(t). \label{U}
\end{align}
Note that \eqref{U} means that, for any function $g\ge0$,
\begin{equation}
  \label{U2}
\intoo g(t)\dd U(t)=\sumn\intoo g(t)\dd F_n(t)
=\sumn \E g(S_n).
\end{equation}

We also allow the summation to start with an initial random variable
$X_0$, which is independent of $X_1,X_2,\dots$, but may have an
arbitrary real-valued distribution. We then define
\begin{align}
\hS_n&\=\sumn X_i=X_0+\sumni X_i, \label{hS}
\\
  \hnu(t)&\=\min\set{n:\hS_n> t}. \label{hnu}
\end{align}

\section{Insertion depth in a trie}\label{Strie}

A \emph{trie} is a binary tree structure designed to store a set of
strings. It is constructed from the strings by the following recursive
procedure, see further \eg{} 
\citetq{Section 6.3}{KnuthIII},
\citetq{Chapter 5}{Mahmoud}
or \citetq{Section 1.1}{Szp-book}: 
If the set of strings is empty, then the trie is empty; 
if there is only one string, then the trie consists of a single node (the
root), and the string is stored there; 
if there is more than one string, then the trie begins with a root,
without any string stored, all strings that begin with 0 are passed to
the left subtree of the root, and all strings that begin with 1 are
passed to the right subtree. In the latter case, the subtrees are constructed
recursively by the same procedure, with the only difference that at
the $k$th level, the strings are partitioned according to the $k$th
letter.
We assume that the strings are distinct (in our random model, this
holds with probability 1), and then the procedure terminates.
Note that one string is stored in each leaf of the trie, and that no
strings are stored in the remaining nodes. The leaves are also called
\emph{external nodes} and the remaining nodes 
are called \emph{internal nodes}; note that every internal node has
one or two children.

The trie is a finite subtree of the complete infinite binary tree $\ctoo$,
where the nodes can be labelled by finite strings
$\aaa=\ga_1\dotsm\ga_k\in\cAx\=\bigcup_{k=0}^\infty\cA^k$ 
(the root is the empty
string).
It is easily seen that a node $\ga_1\dotsm\ga_k$ in $\ctoo$ is an
internal node of the trie if and only if there are at least 2 strings
(in the given set) that start with $\ga_1\dotsm\ga_k$, and (for
$k\ge1$)
that $\ga_1\dotsm\ga_k$ is an external node if and only if there is exactly
one such string, and there is at least one other string beginning with
$\ga_1\dotsm\ga_{k-1}$.

Let $D_n$ be the depth (= path length) of the node containing a given string, 
for example the first, 
in the trie constructed from  $n$ random strings
$\Xi\sss1,\dots,\Xi\sssn$. 
(By symmetry, any of $n$ strings will have a depth with the same
distribution.) Denoting the chosen string by $\Xi=\xi_1\xi_2\dotsm$,
the depth $D_n$ is thus at most $k$ if and only if no other of the
strings begins with $\xi_1\dotsm\xi_k$.
Conditioning on the string $\Xi$, each of the other strings has this
beginning with probability $P(\xi_1\dotsm\xi_k)$, and thus by independence,
recalling \eqref{sop},
\begin{equation}\label{sw1}
  \P(D_n\le k\mid\Xi)
=\bigpar{1-P(\xi_1\dotsm\xi_k)}^{n-1}
=\bigpar{1-e^{-S_k}}^{n-1}.
\end{equation}

Let $X_0=X_0\sssn$ be a random variable with the distribution
\begin{equation}\label{x0}
  \P(X_0\sssn>x)=\bigpar{1-e^x/n}_+^{n-1}
=\bigpar{1-e^{x-\ln n}}^{n-1}_+,
\qquad x\in\oooo.
\end{equation}
As \ntoo, this converges to $\exp(-e^x)$, and thus $X_0\sssn\to
\xxo$, where $-\xxo$ has the Gumbel distribution with $\P(-\xxo\le
x)=\exp(-\exp(-x))$. 

\begin{remark}
It is easily seen that $X_0\sssn\eqd \ln n-\max\set{Z_1,\dots,Z_{n-1}}$,
where $Z_1,Z_2,\dots$ are \iid{} $\Exp(1)$ random variables. 
Cf.\  \citetq{Example 1.7.2}{LLR}.  
\end{remark}

Using \eqref{x0}, we can rewrite \eqref{sw1} as
\begin{equation}\label{sw2}
  \P(D_n\le k\mid\Xi)
=\P\bigpar{X_0\sssn>\ln n-S_k\mid\Xi}
\end{equation}
and thus, recalling \eqref{hS} and \eqref{hnu},
\begin{equation}\label{sw3}
  \P(D_n\le k)
=\P\bigpar{X_0>\ln n-S_k}
=\P\bigpar{\hS_k>\ln n}
=\P\bigpar{\hnu(\ln n)\le k}.
\end{equation}
Since $k\ge1$ is arbitrary, this shows that 
\begin{equation}\label{dnu}
D_n\eqd\hnu(\ln n).  
\end{equation}

In the case $p=1/2$, $S_k=k\ln 2$ is non-random, and the only
randomness in $\hnu(\ln n)$ comes from $X_0$; in fact, it is easy to
see that $\P(D_n\le k)\to \P(-\xxo\le t)$ if $k\to\infty$ and
$n\to\infty$ along sequences such that $k\ln 2-\ln n\to t\in\oooo$, see
\cite{JR86}, \cite{P86}, \cite[Theorem 5.7]{Mahmoud}, \cite{LP06}.
This result can also be expressed as $\dtv(D_n,\ceil{(\ln n-\xxo)/\ln 2})\to0$
as \ntoo, where $\dtv$ denotes the total variation distance of the
distributions, see \cite[Example 4.5]{SJ175}.

However, if $p\neq1/2$, then each $X_k$ is truly random, which
leads to larger dispersion of $D_n$.
We can apply standard renewal theory theorems, see Theorems
\ref{TA1}--\ref{TA2} and \refR{RA1} in the appendix, and immediately
obtain the following. For other, earlier proofs see 
\citetq{Sections 6.3 and 5.2}{KnuthIII},
\citet{P85,P86} and
\citetq{Section 5.5}{Mahmoud}. 
%
The Markov case is treated by \citet{JS91}, ergodic
strings by \citet{P85}, and a class of general dynamical sources
by \citet{CFV01}.

\begin{theorem}\label{TD1}
For every $p\in(0,1)$, 
  \begin{equation}\label{td1}
	\frac{D_n}{\ln n} \pto \frac1{H},
  \end{equation}
with $H$ the entropy given by \eqref{H}.
Moreover, the convergence holds in every $L^r$, $r<\infty$,
too. Hence, all moments converge in \eqref{td1} and
  \begin{equation}
\E D_n^r \sim H^{-r}(\ln n)^r,
\qquad 0<r<\infty. 
\end{equation}
\end{theorem}

\begin{theorem} \label{TD1+}
More precisely:
    \begin{thmenumerate}
\item
  If\/ $\ln p/\ln q$ is irrational, then, as $n\to\infty$,
  \begin{equation}\label{em1}
\E D_n = \frac{\ln n}{H} + \frac{H_2}{2H^2}+\frac{\gamma}{H}+o(1).
  \end{equation}

\item
  If\/ $\ln p/\ln q$ is rational, 
then, as $n\to\infty$,
  \begin{equation}
\E D_n = \frac{\ln n}{H} + \frac{H_2}{2H^2}+\frac{\gamma}{H}
+\ppsi(\ln n)+o(1),
\ppsidef\ppsitd 
\end{equation}
where $\ppsix(t)$ is a small continuous function, with period $d=\gcdpq$ in
$t$,
given by
\begin{equation}\label{ppsiTD1+}
\ppsix(t)\= 
-\frac 1H \sum_{k\neq0}
{\Gamma(-2\pi\ii k/d)}{} e^{2\pi\ii k t/d}.
\end{equation}
  \end{thmenumerate}
\end{theorem}

\begin{proof}
  The non-arithmetic case \eqref{em1} follows directly from
  \eqref{dnu} and \eqref{ta1+a}; we can replace $X_0\sssn$ by the
  limit $\xxo$, and 
since the Gumbel variable $-\xxo$ has characteristic function 
$\E e^{-\ii t\xxo}=\Gamma(1-\ii t)$,  we have $\E\xxo=\Gamma'(1)=-\gamma$.

In the arithmetic case, we use \eqref{ta1+b}, together with
\refL{Lsofie} which yields
\begin{equation*}
  \E\Bigfrakt{\frac td - \frac{\xxo}d}
=
\frac12-\sum_{k\neq0}
\frac{\Gamma(1-2\pi\ii k/d)}{2\pi k \ii} e^{2\pi\ii k t/d}
=
\frac12+\frac1d\sum_{k\neq0}
\Gamma(-2\pi\ii k/d) e^{2\pi\ii k t/d}.
\qedhere	
\end{equation*}
\end{proof}

\begin{theorem}\label{TD2}
  Suppose that $p\in(0,1)$.  Then, as \ntoo,
  \begin{equation*}
\frac{D_n-H\qw \ln n}{\sqrt{\ln n}} \dto N\Bigpar{0,\frac{\gss}{H^3}},
  \end{equation*}
with $\gss=H_2-H^2=pq(\ln p-\ln q)^2$.
If $p\neq1/2$, then $\gss>0$ and this can be written as
\begin{equation*}
  D_n\sim\AsN\bigpar{H\qw\ln n,H\qwww\gss\ln n}.
\end{equation*}
Moreover, 
\begin{equation*}
 \Var D_n = \frac{\gss}{H^3}\ln n+o(\ln n).
\end{equation*}
\end{theorem}

In the argument above, $X_0$ depends on $n$. This is a nuisance,
although no real problem (see \refR{RA1}). 
An alternative that avoids this problem is to 
Poissonize by considering a random number of strings. In this case
it is simplest to consider $1+\Po(\gl)$ strings, so that a selected
string $\Xi$ is compared to a Poisson number $\Po(\gl)$ other strings, for a
parameter $\gl\to\infty$. Conditioned on $\Xi$, the number of other
strings beginning with $\xi_1\dotsm\xi_k$ then has the Poisson
distribution $\Po(\gl P(\xi_1\dotsm\xi_k))$.  Thus we obtain instead
of \eqref{sw1}, now denoting the depth by $D_\gl$,
\begin{equation*}
  \begin{split}
  \P(D_\gl\le k\mid\Xi)
&=
e^{-\gl P(\xi_1\dotsm\xi_k)}=e^{-\gl e^{-S_k}}
=e^{-e^{-(S_k-\ln\gl)}}
\\
&=
\P(-\xxo< S_k-\ln\gl)
=\P(S_k+\xxo> \ln\gl)
	=\P\bigpar{\hnu(\ln\gl)\le k},
  \end{split}
\end{equation*}
where $X_0\=\xxo$ now is independent of $n$, and consequently
$D_\gl\eqd\hnu(\ln\gl)$. We obtain the same asymptotics as for $D_n$
above, directly from Theorems \ref{TA1}--\ref{TA2}.
It is in this case easy to depoissonize, by noting that $D_n$ is
stochastically monotone in $n$, and derive the results for $D_n$ from
the results for $D_\gl$ by choosing $\gl=n\pm n^{2/3}$; we omit the details.

\section{Imbalance in tries}\label{Simbalance}

\citet{Mahmoud:imbalance} studied the imbalance factor of a string
in a trie, defined as the number of steps to the right minus the
number of steps to the left in the path from the root to the leaf
where the string is stored.
We define
\begin{equation*}
  Y_i\=2\xi_i-1=
  \begin{cases}
	-1, & \xi_1=0,
\\
	+1, & \xi_1=1,
  \end{cases}
\end{equation*}
and denote the corresponding partial sums by $V_k\=\sumik Y_i$.
Thus the imbalance factor $\gD_n$ of the string $\Xi$ in a random
trie with $n$ strings is $V_{D_n}$, with $D_n$ as in \refS{Strie} the
depth of the string.

It follows immediately from \eqref{sw2} that \eqref{sw3} holds also
conditioned on the sequence $(Y_1,Y_2,\dots)$.
As a consequence, for any $k$ and $v$,
\begin{equation*}
  \P(D_n= k\mid V_k=v)
=\P\bigpar{\hnu(\ln n)= k\mid V_k=v},
\end{equation*}
which shows that 
\begin{equation*}
(D_n,\gD_n)
  =(D_n,V_{D_n})
\eqd\bigpar{\hnu(\ln n),V_{\hnu(\ln n)}}.  
\end{equation*}
In particular,
\begin{equation*}
\gD_n
\eqd V_{\hnu(\ln n)}.  
\end{equation*}
We may apply \refT{TAVtaut} (and \refR{RVtaut}).
A simple calculation yields $\Var(\mu_X Y_1-\mu_Y X_1)=pq(\ln p+\ln
q)^2=pq\ln^2 (pq)$, and we obtain
the 
central limit theorem by \citet{Mahmoud}:
\begin{theorem}
  As \ntoo,
  \begin{equation*}
\gD_n
\sim
\AsN\lrpar{\frac{p-q}{H}\ln n,\,\frac{pq\ln^2(pq)}{H^3}\ln n}.	
  \end{equation*}
\end{theorem}

\section{The expected size of a trie}\label{Sexp}

A trie built of $n$ strings as in \refS{Strie} has $n$ external nodes,
since each external node contains exactly one string. 
However, the number of
internal nodes, $\xY_n$, say, is random. 
We will study its expectation.
For simplicity we Poissonize directly and consider a trie constructed
from $\Po(\gl)$ strings; we let $\tY_\gl$ be the number of internal
nodes.
The results below have previously been found by other methods, in
particular, more precise asymptotics have been found using Mellin
transforms; see \citet{KnuthIII}, \citet{Mahmoud}, 
\citet{FFHJ85}, and, in particular, \citet{JR87,JR89}.
The Markov case is studied by \citet{R88}
and dynamical sources by \citet{CFV01}.

If $\aaa=\ga_1\dotsm\ga_k$ is a finite string, let $I(\aaa)$ be the
indicator of the event that $\aaa$ is an internal node in the trie.
We found above that this event occurs if and only if there are at
least two strings beginning with $\aaa$. In our Poisson model, the
number of strings beginning with $\aaa$ has a Poisson distribution
$\Po(\gl P(\aaa))$, and thus
\begin{equation} \label{emma}
  \begin{split}
  \E \tY_\gl = \sumaa \E I(\aaa)
= \sumaa\P\bigpar{\Po(\gl P(\aaa))\ge 2}
=\sumaa	f(\gl P(\aaa)),
  \end{split}
\end{equation}
where
\begin{equation}\label{femma}
  f(x)\=\P\bigpar{\Po( x)\ge 2}
=1-(1+x)e^{- x}.
\end{equation}

Sums of the type in \eqref{emma} are often studied using Mellin
transform inversion and residue calculus. Renewal theory presents an
alternative. As said in the introduction, this opens the way to
straightforward generalizations, \eg{} to Markov sources.

\begin{theorem}\label{Tsum}
  Suppose that $f$ is a non-negative function on  $(0,\infty)$, and that
$F(\gl)=\sumaa	f(\gl P(\aaa))$, with $P(\aaa)$ given by \eqref{paaa}.
Assume further that $f$ is \aex\ continuous 
and satisfies the estimates 
\begin{align}\label{dRie}
f(x)=O(x^2),\quad 0< x<1, \qquad\text{and}\qquad f(x)=O(1), \quad 1<x<\infty.
\end{align}
Let $\ggh(t)\=e^tf(e^{-t})$.
  \begin{thmenumerate}
\item
  If\/ $\ln p/\ln q$ is irrational, then, as $\gl\to\infty$,
  \begin{equation}\label{jeiw}
	\frac{F(\gl)}{\gl} \to 
\frac1 H \intoooo g(t)\dd t
=
\frac1 H \intoo f(x)x\qww\dd x.	
  \end{equation}

\item
  If\/ $\ln p/\ln q$ is rational, 
then, as $\gl\to\infty$,
  \begin{equation}\label{jeiiw}
	\frac{F(\gl)}{\gl} = \frac1 H \psi(\ln\gl) + o(1),
  \end{equation}
where,
with $d\=\gcdpq$ given by \eqref{d}, 
$\psi$ is a bounded $d$-periodic function having the Fourier series
\begin{equation}\label{tsumfou}
  \psi(t)\sim\summoooo\hpsi(m)e^{2\pi\ii mt/d}
\end{equation}
with
\begin{equation}\label{tsumfouf}
  \hpsi(m)=\widehat \ggh(-2\pi m/d) = \intoooo e^{2\pi\ii m t/d}\ggh(t)\dd t
=\intoo f(x)x^{-2-2\pi\ii m/d}\dd x.
\end{equation}
Furthermore, 
\begin{equation}
  \label{tsumsum}
\psi(t)=d\sumkoooo g(kd-t).
\end{equation}
If $f$ is continuous, then $\psi$ is too.
  \end{thmenumerate}
\end{theorem}

\begin{proof}
If $\FF(\aaa)$ is any non-negative function on $\cAx$, then, using
\eqref{sop}, for each $k\ge0$, 
\begin{equation*}
  \begin{split}
  \sum_{\ga_1,\dots,\ga_k} \FF(\ga_1\dotsm\ga_k)
&=
  \sum_{\ga_1,\dots,\ga_k}
\frac{ \FF(\ga_1\dotsm\ga_k)}{P(\ga_1\dotsm\ga_k)}{P(\ga_1\dotsm\ga_k)}
\\&
=  \E \frac{ \FF(\xi_1\dotsm\xi_k)}{P(\xi_1\dotsm\xi_k)}
=  \E \bigpar{e^{S_k} \FF(\xi_1\dotsm\xi_k)},	
  \end{split}
\end{equation*}
and thus, 
\begin{equation}\label{FF}
  \sumaa \FF(\aaa)=\sumk \E\bigpar{ e^{S_k} \FF(\xi_1\dotsm\xi_k)}.
\end{equation}

With $\FF(\aaa)=f(\gl P(\aaa))$, we have
$\FF(\xi_1\dotsm\xi_k)=f(\gl e^{-S_k})$ and thus \eqref{FF} yields,
recalling \eqref{U}, 
\begin{equation*}
F(\gl)=  \sumaa f(\gl P(\aaa))
=\sumk \E\bigpar{ e^{S_k} f(\gl e^{-S_k})}
=\intoo f(\gl e^{-x})e^x\dd U(x).
\end{equation*}
Define further $\ggx(x)\=f(x)/x$; thus
$\ggh(t)=\ggx(e^{-t})$. Then, 
\begin{equation}\label{magnus}
F(\gl)
=\intoo \gl \ggx(\gl e^{-x})\dd U(x)
=\gl \intoo \ggh(x-\ln \gl)\dd U(x).
\end{equation}

We can now apply the key renewal theorem, \refT{TAkey}.
The function $g$ is \aex\ continuous and it follows from \eqref{dRie} that 
$\ggh(t)\le C e^{-|t|}$ for some $C$; hence $g$ is directly
Riemann integrable on $\oooo$ by \refL{LRiemann}.
In the non-arithmetic case (i) we obtain \eqref{jeiw} from 
\eqref{magnus} and
\eqref{key+},
since $\mu=\E X_i=H$ by \eqref{H} and,
with $x=e^{-t}$,
\begin{equation}\label{sigw}
  \intoooo \ggh(t)\dd t
= \intoooo e^t f(e^{-t})\dd t
= \intoo f(x)x\qww\dd x.
\end{equation}

Similarly, the aritmetic case (ii)
follows from
\eqref{keya+} and \eqref{psi}--\eqref{fouf} together with the
calculation, generalizing \eqref{sigw}, 
\begin{equation*}
\hggh(s)
=  \intoooo e^{-\ii st} \ggh(t)\dd t
= \intoooo e^{(1-\ii s)t} f(e^{-t})\dd t
= \intoo f(x)x^{-2+\ii s}\dd x.
\end{equation*}
(This equals the Mellin transform $\widetilde f(-1+\ii s)$.)
\end{proof}

\begin{remark}
  The assumptions on $f$ may be weakened (with the same proof); it
  suffices that 
$f(x)= O( x^{1-\gd})$ and $f(x)= O( x^{1+\gd})$ for 
$x\in(0,\infty)$
  and some $\gd>0$. If $f$ is continuous, it is obviously sufficient
  that these estimates hold for small and large $x$, respectively.
\end{remark}

Returning to $\tY_\gl$, we obtain the following for the expected
number of internal nodes in the Poisson trie.
\begin{theorem}\label{Ttrie}
  \begin{thmenumerate}
\item
  If\/ $\ln p/\ln q$ is irrational, then, as $\gl\to\infty$,
  \begin{equation}\label{jei}
	\frac{\E \tY_\gl}{\gl} \to \frac1 H.	
  \end{equation}

\item
  If\/ $\ln p/\ln q$ is rational, 
then, as $\gl\to\infty$,
  \begin{equation}\label{jeii}
	\frac{\E \tY_\gl}{\gl} = \frac1H +\frac1 H \ppsi(\ln\gl) + o(1),
  \end{equation}
where,
with $d=\gcdpq$,
$\ppsix\ppsidef\ppsitrie$ 
is a continuous $d$-periodic function with average $0$ and Fourier
expansion
\begin{equation*}
  \ppsix(t)
=\sum_{k\neq0}\frac{\Gamma(1-2\pi\ii k/d)}{1+2\pi\ii k/d}
 e^{2\pi\ii k t/d}
=\sum_{k\neq0}\frac{2\pi\ii k}{d}{\Gamma\Bigpar{-1-\frac{2\pi\ii k}{d}}}
 e^{2\pi\ii k t/d}.
\end{equation*}
  \end{thmenumerate}
\end{theorem}

\begin{proof}
We apply \refT{Tsum} to \eqref{emma}.
It follows from \eqref{femma} that $f'(x)=xe^{-x}$.
Thus, by an integration by parts,
since $f(x)/x\to0$ as $x\to0$ and $x\to\infty$, 
\begin{equation}\label{sig}
 \intoo f(x)x\qww\dd x
= \intoo f'(x)x\qw\dd x
= \intoo e^{-x}\dd x
=1.
\end{equation}
Consequently,
\eqref{jei} follows from \eqref{jeiw}.

Similarly,
\eqref{jeii} follows from
\eqref{jeiiw}, and the calculation, generalizing \eqref{sig},
\begin{align*}
\hggh(s)
&
= \intoo f(x)x^{-2+\ii s}\dd x
= (1-\ii s)\qw\intoo f'(x)x^{-1+\ii s}\dd x
\\&
= \frac{\Gamma(1+\ii s)}{1-\ii s}
= -\ii s\Gamma(-1+\ii s).
\qedhere
\end{align*}
\end{proof}

The case of a fixed number $n$ of strings is easily handled by
comparison, and \eqref{jei} and \eqref{jeii} imply the corresponding results
for $\xY_n$:

\begin{theorem}\label{Ttrien}
  \begin{thmenumerate}
\item
  If\/ $\ln p/\ln q$ is irrational, then, as $n\to\infty$,
  \begin{equation*}
	\frac{\E \xY_n}{n} \to \frac1 H.	
  \end{equation*}

\item
  If\/ $\ln p/\ln q$ is rational, 
then, as $n\to\infty$,
with $\ppsitrie$ as in \refT{Ttrie},
  \begin{equation*}
	\frac{\E \xY_n}{n} = \frac1H+\frac1 H \ppsitrie(\ln n) + o(1).
  \end{equation*}
  \end{thmenumerate}
\end{theorem}

\begin{proof}
   $\E \xY_n$ is increasing in $n$.
Thus, first, because $\P(\Po(2n)\ge n)\ge1/2$,
$\E \tY_{2n}\ge \frac12 \E \xY_n$, and thus 
$\E \xY_n\le 2 \E \tY_{2n} = O(n)$.
Secondly, using this estimate, the standard Chernoff concentration
bounds for the Poisson distribution easily implies, 
with $\gl_{\pm}=n\pm
n^{2/3}$, say,
$\E \tY_{\gl_-} +o(n) \le \E \xY_n \le \E \tY_{\gl_+} +o(n)$.
The results then follow from \refT{Ttrie}.
\end{proof}

\begin{remark}
  It is well-known that the periodic function $\ppsitrie$ above, as in many
  similar results, fluctuates very little from its mean. In fact, the
  largest $d$ is obtained for $p=q=1/2$, when $d=\ln 2$. Since
  $\Gamma(1+\ii s)$ decreases rapidly as $s\to\pm\infty$, the Fourier
coefficients of $\ppsitrie(t)$ are very small; the largest (in
  absolute value) are 
  $|\widehat\ppsitrie(\pm1)|=|\Gamma(1+2\pi\ii/\ln2)|/|1-2\pi\ii/\ln
  2|\approx0.542\cdot 10^{-6}$, so 
$|\ppsitrie(\ln n)|$ is at most   about $10^{-6}$, and
the oscillations 
$\ppsitrie(\ln n)/H$ 
of $\E W_n/n$ are bounded by $1.6\cdot10^{-6}$.
(See for example \cite[pp.~23--28]{Mahmoud}.)
Other choices of $p$ yield even smaller oscillations.
\end{remark}

\section{$b$-tries}\label{Sb}
  As a variation, consider a $b$-trie, where each node can store $b$
  strings, for some fixed integer $b\ge1$; as before, the internal
  nodes do not contain any string. A finite string $\aaa$ now is an
  internal node if and only if at least $b+1$ of the strings start
  with $\aaa$. In the argument above we only have to replace
  \eqref{femma} by
  \begin{equation}
f(x)\=\P\bigpar{\Po(x)\ge b+1}; 
  \end{equation}
thus  
$f'(x)=\P\bigpar{\Po(x)= b}=x^be^{-x}/b!$
and \eqref{sigw} yields, with an integration by parts as in \eqref{sig}, 
$\intoooo \ggh(t)\dd t = 1/b$.
Hence, in the non-arithmetic case when $\ln p/\ln q$ is irrational,
the expected number of internal nodes is $\E \tY\sss b_\gl \sim \gl
/(H b)$, as found by \citet{JR87,JR89}. In the arithmetic case, we
obtain a periodic function $\psi$, now with Fourier coefficients
$(1+2\pi\ii k/d)\qw\Gamma(b-2\pi\ii k/d)/b!$. 

We can also analyze the external nodes. Let $\zj$ be the number of
nodes where exactly $j$ strings are stored, $j=1,\dots,b$.
A finite string $\aaa$ is one of these nodes if exactly $j$ of the
stored strings begin with $\aaa$, and at least $b-j+1$ other strings
begin with $\aaa'$, the sibling of $\aaa$ obtained by flipping the
last letter. (We assume that there are at least $b$ strings, so we can
ignore the root.)

Consider again the Poisson model.
In the case when $\aaa$ ends with 1, \ie, $\aaa=\bbb1$ for some
$\bbb$, the probability of this event is, with $x=\gl P(\bbb)$,
by independence in the Poisson model,
$\P\bigpar{\Po(px)=j} \P\bigpar{\Po(qx)>b-j}$. If $\aaa=\bbb0$, we
similarly have the probability
$\P\bigpar{\Po(qx)=j} \P\bigpar{\Po(px)>b-j}$. Summing over $\beta\in\cAx$, 
we thus obtain a sum of the type in \refT{Tsum} with $f$ replaced by 
\begin{equation*}
  \begin{split}
  f_j(x)&=
\P\bigpar{\Po(px)=j} \P\bigpar{\Po(qx)>b-j}
+ \P\bigpar{\Po(qx)=j} \P\bigpar{\Po(px)>b-j}
\\
&=
\frac{p^jx^j}{j!}e^{-px}\lrpar{1  -\sum_{k=0}^{b-j}\frac{q^kx^{k}}{k!}e^{-qx}}
+\frac{q^jx^j}{j!}e^{-qx}\lrpar{1 -\sum_{k=0}^{b-j}\frac{p^kx^{k}}{k!}e^{-px}}
\\
&=
\frac{p^jx^j}{j!}e^{-px} +\frac{q^jx^j}{j!}e^{-qx} 
-\sum_{k=0}^{b-j}\frac{(p^jq^k+q^jp^k)x^{j+k}}{j!\,k!}e^{-x}
.	
  \end{split}
\end{equation*}

We argue as above, with $\ggh_j(t)\=e^tf_j(e^{-t})$. We have,
similarly to \eqref{sigw},
omitting some details,
\begin{equation}
  \label{cj1}
  \begin{split}
c_j&\=\intoooo\ggh_j(t)\dd t=\intoo f_j(x)x\qww\dd x
\\&\phantom:=
\begin{cases}
p\ln(1/p)+q\ln(1/q)
-\sum_{k=1}^{b-1}\frac1{k}(pq^k+qp^k),
  & j=1,\\
\frac{1}{j(j-1)}
-\sum_{k=0}^{b-j}\frac{(j+k-2)!}{j!\,k!}(p^jq^k+q^jp^k),
& 2\le j\le b.
\end{cases}	
  \end{split}
\end{equation}

Alternatively, using
\begin{equation*}
  \begin{split}
  f_j(x)&=
\frac{p^jx^j}{j!}e^{-px}\sum_{k=b-j+1}^\infty\frac{q^kx^{k}}{k!}e^{-qx}
+\frac{q^jx^j}{j!}e^{-qx}\sum_{k=b-j+1}^{\infty}\frac{p^kx^{k}}{k!}e^{-px}
\\
&=
\sum_{k=b-j+1}^{\infty}\frac{(p^jq^k+q^jp^k)x^{j+k}}{j!\,k!}e^{-x}
,	
  \end{split}
\end{equation*}
we find 
\begin{equation}
  \label{cj2}
  \begin{split}
c_j&=
\sum_{k=b-j+1}^\infty\frac{(j+k-2)!}{j!\,k!}(p^jq^k+q^jp^k),
\qquad 1\le j\le b.
  \end{split}
\end{equation}
More generally (except when $(j,s)=(1,0)$),
\begin{equation}
  \label{hhj}
  \begin{split}
\hggh_j(s)&=\intoo f_j(x)x^{-2+\ii s}\dd x
\\&=
\frac{\Gamma(j-1+\ii s)}{j!}(p^{1-\ii s}+q^{1-\ii s})
-\sum_{k=0}^{b-j}\frac{\Gamma(j+k-1+\ii s)}{j!\,k!}(p^jq^k+q^jp^k).
  \end{split}
\end{equation}

If we use the notation $\zjn$ for the trie with a fixed number $n$ of
strings and $\tzj$ for the Poisson model with $\Po(\gl)$ strings, we
obtain as above the following result 
for the number of external nodes that store $j$ strings.

\begin{theorem}\label{Tbtrien}
  \begin{thmenumerate}
\item
  If\/ $\ln p/\ln q$ is irrational, then, as $n\to\infty$,
for $j=1,\dots,b$,
  \begin{equation*}
	\frac{\E \zjn}{n} \to \pi_j\=\frac{c_j} H,
  \end{equation*}
with $c_j$ given by \eqref{cj1}--\eqref{cj2}.

\item
  If\/ $\ln p/\ln q$ is rational, 
then, as $n\to\infty$, for $j=1,\dots,b$,
  \begin{equation*}
	\frac{\E \zjn}{n} =  \psibj(\ln n) + o(1),
  \end{equation*}
where $\psibj$ is a continuous $d$-periodic function,
with $d$ as in \refT{Ttrie}; $\psibj$ has
average $\pi_j$ and Fourier expansion
\begin{equation*}
  \psibj(t)
=H\qw\sumkoooo\hggh_j(-2\pi\ii k/d) e^{2\pi\ii k t/d}
=\pi_j+H\qw\sum_{k\neq0}\hggh_j(-2\pi\ii k/d) e^{2\pi\ii k t/d},
\end{equation*}
with $\hggh_j$ given by \eqref{hhj}.
  \end{thmenumerate}
The same results (with $n$ replaced by $\gl$) hold for $\tzj$ in the
Poisson model. 
\end{theorem}

\begin{proof}
As just said, the Poisson case follows from \eqref{Tsum}, and 
  it remains only to depoissonize. To do this, choose $\gl=n$, and let
  $N\sim\Po(n)$ be the number of strings in the Poisson model. We
  couple the trie with $n$ strings and the Poisson trie with $N$
  strings by starting with $\min(n,N)$ common strings.
If we add a new string to the trie, it is either stored in an existing
  leaf or it converts a leaf to an internal node and adds two new
  leafs (and possibly a chain of further internal nodes). Thus at most
  3 leaves are affected, and each $Z_j$ changes by at most 3. Since we
  add $\max(n,N)-\min(n,N)=|N-n|$ new strings, we have $|\tzj-\zjn|\le
  3|N-n|$ for each $j$, and thus
$|\E\tzj-\E\zjn|\le 3\E |N-n|=O(\sqrt n)$.
\end{proof}

For example, for $b=2,3,4$ we have the following limits in the
non-aritmetic case, and up to small oscillations also in the aritmetic case:
$$
\begin{tabular}{c|c|c|c|c}
$b$ & $\pi_1$ & $\pi_2$ & $\pi_3$ & $\pi_4$ \\
\hline
2 & $1-\frac{2}H pq$ & $\frac1H{pq}$& \rule{0pt}{12pt} \\[3pt]
3 & $1-\frac{5}{2H}pq$ & $\frac{1}{2H}pq$& $\frac{1}{2H}pq$ \\[3pt]
4 & $1-\frac{17}{6H}pq+\frac{2}{3H}(pq)^2$ &
  $\frac1{2H}pq-\frac1H(pq)^2$ & 
  $\frac1{6H}pq+\frac2{3H}(pq)^2$ & 
  $\frac1{3H}pq-\frac1{6H}(pq)^2$\!\!\!\!\\  
\end{tabular}
$$
Note that $\sum_{1}^b j\pi_j=1$, 
or equivalently $\sum_{1}^b jc_j=H$, 
since the total number of strings in
the leaves is $n$; this can also be verified from \eqref{cj1}.

\section{Patricia tries}\label{Spat}

Another version of the trie is the Patricia trie, where the trie is
compressed by 
eliminating all internal nodes with only one child. 
(We use the notations above with a superscript $P$ for the Patricia case.)
Since each internal node in the Patricia trie thus has exactly 2
children, the  number of internal nodes is one less than the number of
external nodes, \ie{} $W_n\pat=n-1$ for a Patricia trie with $n$ strings.

As another illustration of \refT{Tsum}, we note that this trivial result,
to the first order at least, also can be derived as above.
The condition for
a finite string $\aaa$ to be an internal node of the Patricia trie is
 that there is at least one string beginning with $\aaa0$ and at
least one string beginning with $\aaa1$. In the Poisson model, the
number of strings with these beginnings are independent Poisson random
variables with means $\gl P(\aaa0)=\gl qP(\aaa)$ and $\gl P(\aaa1)=\gl
pP(\aaa)$, and we can argue as above with $f(x)=(1-e^{-px})(1-e^{-qx})$.
In this case, $\intoooo \ggh(t)\dd t=\intoo f(x)x^{-2}=-p\ln p-q\ln q=H$,
which implies $\E \tY_\gl\pat\sim \gl$ and
$\E \xY_n\pat\sim n$ in the non-arithmetic
case. Moreover, we know that this holds in the arithmetic case too,
without oscillations, which means that $\hpsi(m)=0$ for $m\neq0$ in
\eqref{tsumfou}--\eqref{tsumfouf}. 
Indeed, for example by integration by parts,
\begin{equation*}
  \begin{split}
\hggh(s)&=  \intoo f(x) x^{-2+\ii s}\dd x =
\intoo  x^{-2+\ii s} (1-e^{-px}-e^{-qx}+e^{-x})\dd x 
\\&
=
\bigpar{1-p^{1-\ii s}-q^{1-\ii s}}\Gamma(-1+\ii s),	
  \end{split}
\end{equation*}
and thus $\hpsi(m)=\hggh(-2\pi m/d)=0$ for $m\neq0$.

We can also consider a Patricia $b$-trie, and obtain the asymptotics
of the expected number of internal nodes in a similar way, but it is
simpler to use the result in \refT{Tbtrien} and the fact that the
number of internal nodes is 
$\sum_{j=1}^b Z_{jn}\pat-1=\sum_{j=1}^b Z_{jn}-1$; 
in the
non-arithmetic case this yields the asymptotics 
$\bigpar{\sum_{j=1}^b\pi_j}n$.

The number of internal nodes in the Patricia trie
is reduced to $n-1$ from about $n/H$ in the trie (see \refT{Ttrien}, and ignore
the small oscillations in the arithmetic case); this is a
reduction by a factor $H$ which is at most $\ln2\approx0.693$, in
other words a reduction with at least 30\%.
Nevertheless, the reduction in the path length to a given string is
negligible.
In fact, if we for simplicity, as in \refS{Strie}, consider
$1+\Po(\gl)$ strings, with one selected string $\Xi$, then a string
$\aaa$ is an internal node on the path in the trie from the root to $\Xi$ such
that $\aaa$ does not appear in the Patricia trie if and only if $\Xi$ begins
with $\aaa$, and further, either $\Xi$ begins with $\aaa0$, there is
at least one other such string, and there is no string beginning with
$\aaa1$,
or, conversely, $\Xi$ and at least one other string begins with
$\aaa1$ but no string begins with $\aaa0$.
The probability of this is $\gl\qw f(x)$ with $x=\gl P(\aaa)$ and
\begin{equation*}
  f(x)\=x q (1-e^{-qx})e^{-px}+ x p (1-e^{-px})e^{-qx}.
\end{equation*}
Hence, if $\Delta D_\gl\=D_\gl-D_\gl\pat$ is difference between the path
lengths to $\Xi$ in the 
trie and in the Patricia trie, 
then
$\E \Delta D_\gl=\gl\qw\sum_{\aaa} f(\gl P(\aaa))$ and \refT{Tsum} yields
\begin{equation*}
  \begin{split}
  \E \Delta D_\gl &\to \frac1H\intoooo f(x)x\qww\dd x
\\&
=
\frac qH\intoooo\frac{e^{-px}-e^{-x}}{x}\dd x
+\frac pH\intoooo\frac{e^{-qx}-e^{-x}}{x}\dd x
\\&
=\frac{-q\ln p-p\ln q}{H}.	
  \end{split}
\end{equation*}
This holds also in the arithmetic case, since a simple calculation
shows that Fourier coefficients $\hpsi(m)$ in \eqref{tsumfouf}
vanish for all $m\neq0$.
(This is an interesting example of cancellation in an arithmetic case
where we would expect oscillations.)
Hence the expected saving is 1 for $p=1/2$, and $O(1)$ for any fixed
$p$. (This is $o(\E D_\gl)$ and thus asymptotically negligible.)

Again, we can depoissonize by considering $\gl=n\pm n\qqqb$, and we
obtain the same result for a fixed number $n$ of strings.
Together with \refT{TD1+}, we obtain the following, earlier found by
\citet{Szp90}, 
see also \citetq{Section 6.3}{KnuthIII} ($p=1/2$) and
\citet{RJSz:Patricia}. 
(Dynamical sources are considered by \citet{Bourdon01}.)

\begin{theorem} For the expected depth $\E D^P_n$ in a Patricia trie:
    \begin{thmenumerate}
\item
  If\/ $\ln p/\ln q$ is irrational, then, as $n\to\infty$,
  \begin{equation*}
\E D^P_n = 
\frac{\ln n}{H} + \frac{H_2}{2H^2}+\frac{\gamma+q\ln p+p\ln q}{H}+o(1). 
  \end{equation*}

\item
  If\/ $\ln p/\ln q$ is rational, 
then, as $n\to\infty$,
  \begin{equation*}
\E D^P_n = \frac{\ln n}{H} + \frac{H_2}{2H^2}+\frac{\gamma
+q\ln p+p\ln q}{H}
+\ppsitd(\ln n)+o(1),
\end{equation*}
where $\ppsitd(t)$ is a small continuous function, with period $d$ in
$t$,
given by \eqref{ppsiTD1+}.
  \end{thmenumerate}
\end{theorem}

\section{Insertion in a trie} \label{Sins}

When a new string is inserted in a trie, it becomes a new external
node; it may also create one or several new internal nodes. Let
$N\ge0$ be the number of new internal nodes.

\begin{theorem}
  \label{Tnew}
As \ntoo,
\begin{align*}
  \P(N=0)&=1-\frac{2pq}H-\ppsi(\ln n) +o(1),\\
\P(N=j)&=\Bigpar{\frac{2pq}H+\ppsix(\ln n)}2pq(1-2pq)^{j-1}+o(1),
\qquad j\ge1,
\end{align*}
where $\ppsix=0$ in the non-arithmetic case, while in the $d$-arithmetic
case
\begin{equation*}
  \ppsix(t)=\frac{2pq}H\sum_{k\neq0} \Gamma\bigpar{1-\frac{2\pi\ii k}d}
  e^{2\pi\ii kt/d}.
\end{equation*}
Further,
\begin{equation}
  \label{qe}
\E N=\frac1H+\frac1{2 pq}\ppsix(\ln n)+o(1).
\end{equation}
The same results hold in the Poisson case (with $n$ replaced by $\gl$).
\end{theorem}

\begin{proof}
  Consider first the Poisson case, with insertion of $\Xi$ in a trie
  with $\Po(\gl)$ other strings.

Let $K$ be the length of the longest prefix of $\Xi$ that is shared
with at least two strings already existing in the trie; this is the
depth of the last internal node (in the existing trie) that the new
string encounters while being inserted.

There is either no existing string with the same $K+1$ first letters
as $\Xi$, or exactly one such string. In the first case, $\Xi$ is
inserted at depth $K+1$ without creating any new internal nodes, so
$N=0$.

In the second case, we have reached an external node, which is
converted into an internal node, and the string that was stored there
is displaced and instead stored, together with the new string, at the
end of a sequence of $N\ge1$ new internal nodes, where $N$ is the
number of common letters, after the $K$ first, in these two strings.

Thus, conditioned on $N\ge1$, $N$ has a geometric distribution:
\begin{equation}
  \label{q3}
\P(N=j)=\P(N\ge1)(p^2+q^2)^{j-1}\cdot2pq,\qquad j\ge1.
\end{equation}
Since further $\P(N=0)=1-\P(N\ge1)$, it suffices to find $\P(N\ge1)$.

For a given $k$, the event $N\ge1$, $K=k$ and, say, $\xi_{K+1}=1$,
happens if and only if $\xi_{k+1}=1$ and there is exactly one
existing string beginning with $\xi_1\dotsm\xi_k1$ and at least one
beginning with $\xi_1\dotsm\xi_k0$.
The conditional probability of this given $\ga\=\xi_1\dotsm\xi_k$ is
\begin{equation*}
  \P(\xi_{k+1}=1)\P\bigpar{\Po(\gl P(\ga)q)\ge1}\P\bigpar{\Po(\gl P(\ga)p)=1}
=f_1\bigpar{\gl P(\ga)},
\end{equation*}
with
\begin{equation*}
  f_1(x)=p(1-e^{qx})(pxe^{-px})
=p^2xe^{-px}-p^2xe^{-x}.
\end{equation*}
Thus,
\begin{equation*}
  \begin{split}
  \P(N\ge1,\, K=k \text{ and } \xi_{K+1}=1)
&=\E f_1\bigpar{\gl \P(\xi_1\dotsm\xi_k)}
=\E f_1\bigpar{\gl e^{-S_k}}
\\&
=\E f_1\bigpar{e^{-(S_k-\ln\gl)}}	
  \end{split}
\end{equation*}
and, summing over $k$ and using \eqref{U2},
\begin{equation*}
  \P(N\ge1 \text{ and } \xi_{K+1}=1)
=\sumk\E f_1\bigpar{e^{-(S_k-\ln\gl)}}
=\intoo f_1\bigpar{e^{-(x-\ln\gl)}}\dd U(x).
\end{equation*}
The function $g_1(x)\=f_1(e^{-x})$  is directly
Riemann integrable on $\oooo$ by \refL{LRiemann}
(because $f_1(x)=O(x\bmin x\qw)$), and thus the key renewal theorem
\refT{TAkey} yields
\begin{equation}\label{q4a}
  \P(N\ge1 \text{ and } \xi_{K+1}=1)
=\frac1H \intoooo g_1(x)\dd x + \ppsiq1(\ln\gl)+o(1).
\end{equation}
where $\ppsiq1(t)=0$ in the non-arithmetic case and 
\begin{equation}
  \label{q4}
\ppsiq1(t)=\frac1H \sum_{m\neq0}\hg_1(-2\pi m/d)e^{2\pi\ii mt/d}
\end{equation}
in the arithmetic case.

Routine integrations yield
\begin{equation}
  \label{qk1}
\intoooo g_1(x)\dd x = \intoo f_1(y)\frac{\dd y}y
= \intoo (p^2 e^{-py}-p^2e^{-y})\dd y
=p-p^2=pq
\end{equation}
and, more generally,
\begin{equation*}
  \begin{split}
\hg_1(s)=
\intoooo e^{-\ii sx}g_1(x)\dd x = \intoo f_1(y)y^{\ii s-1}\dd y
=(p^{1-\ii s}-p^2)\Gamma(1+\ii s)	;
  \end{split}
\end{equation*}
thus in the arithmetic case, since $p^{2\pi\ii m/d}=1$ for integers $m$,
\begin{equation}
  \label{qk2}
\hg_1(-2\pi m/d)
=pq\Gamma(1-2\pi m\ii /d).
\end{equation}

By symmetry, \eqref{q4a} implies, for similarly defined $g_0$ and $\psi_0$, 
\begin{equation}\label{qk2a}
  \P(N\ge1 \text{ and } \xi_{K+1}=0)
=\frac1H \intoooo g_0(x)\dd x + \ppsiq0(\ln\gl)+o(1),
\end{equation}
where, noting that \eqref{qk1} and \eqref{qk2} are symmetric in $p$
and $q$, $\intoooo g_0(x)\dd x=pq$ and $\ppsiq0=\ppsiq1$.

Consequently, summing \eqref{q4a} and \eqref{qk2a},
with $\ppsix\=\ppsiq0+\ppsiq1=2\ppsiq1$,
\begin{equation}\label{qk3}
  \P(N\ge1)
=\frac{2pq}H + \ppsix(\ln\gl)+o(1).
\end{equation}
The result in the Poisson case now follows from \eqref{q3}, \eqref{q4},
\eqref{qk2} and \eqref{qk3}.
For the mean we have by \eqref{q3} and \eqref{qk3},
\begin{equation*}
  \E N=\sumj j\P(N=j) =\frac1{2pq}\P(N\ge1)
=\frac{1}H + \frac1{2pq}\ppsix(\ln\gl)+o(1).
\end{equation*}

To depoissonize, consider first adding $\Xi$ to a trie with
$\Po(n-n\qqqb)$ strings, and then increase the family by adding 
$\Po(n\qqqb)$ further strings; it is easily seen that with probability
$1-O(\gl\qqqw)=1-o(1)$, this does not change the place where $\Xi$ is
inserted, and thus not $N$. The same holds for all intermediate tries,
in particular for the one with exactly $n$ strings if there is one,
which there is \whp\ because $\P\bigpar{\Po(n-n\qqqb)\le n}\to1$ and
$\P\bigpar{\Po(n+n\qqqb)\ge n}\to1$. Hence the variable $N$ is \whp\
the same for $n$ strings and for $\Po(n)$ strings.
\end{proof}

It is easily verified that, 
at least if we ignore the error terms,
the expected number of new internal nodes added for each new string
given by \eqref{qe}
coincides with the derivative of 
$\E W_\gl=\frac\gl H+\frac\gl H \ppsitrie(\ln\gl)+o(\gl)$
given by \eqref{jeii}, as it should.

\begin{remark}
  \citet{ChrMahmoud} studied random climbing in random tries,
taking (in one version) steps left or  right with probabilities $p$
and $q$; this is 
like inserting a new node but without moving any old one. The length
of the climb is thus $D_n$ when $N=0$ or 1 but 
$D_n-(N-1)$ when $N\ge1$.

The average climb length found by \citet{ChrMahmoud} for this version
thus follows from
Theorems \refand{TD1+}{Tnew}.
\end{remark}

\section{Tunstall and Khodak codes}\label{Scodes}

Tunstall and Khodak codes are variable-to-fixed length codes that are
used in data compression. We give a brief description here.
See \cite{DRSS06}, \cite{DRSS08} and the survey \cite{Szp-codes} 
for more details
and references, as well as for an analysis using Mellin transforms.

We recall first the general situation. The idea is that an infinite 
string can be parsed as a unique sequence of nonoverlapping 
\emph{phrases} belonging to a certain (finite) \emph{dictionary} $\cD$.
Each phrase in the dictionary then can be represented by a binary number
of fixed length $\ell$; if there are $M$ phrases in the dictionary we take
$\ell\=\ceil{\lg M}$. 

Note first that a set of phrases is a dictionary allowing a unique
parsing in the way just described
if and only if every infinite string has exactly one prefix in
the dictionary. Equivalently, the phrases in the dictionary have to be
the external nodes of a trie where every internal node has two
children (so the Patricia trie is the same); this trie is the parsing tree.

By a random phrase we mean a phrase distributed as the unique initial
phrase in a random infinite string $\Xi$. Thus a phrase $\aaa$ in the
dictionary $\cD$ is chosen with probability $P(\aaa)$.
We let the random variable $\dl$ be the length of a random phrase.

If we parse an infinite \iid{} string $\Xi$, the successive phrases
will be independent with this distributions. Hence, if $K_N$ is the
(random) number of phrases required to code the $N$ first letters
$\xi_1\dotsm\xi_N$, then, see \refApp{App} and \eqref{nu},
$K_N=\nu(N-1)$ for a renewal process where the increments $X_i$ are
independent copies of $\dl$. Consequently, as $N\to\infty$, by
\refT{TA1}, 
\begin{equation}
\frac{K_N}N\asto \frac1{\E \dl}
\qquad\text{and}\qquad
\frac{\E K_N}N\to \frac1{\E \dl}.
\end{equation}
We obtain also convergence of higher moments and, by \refT{TA2}, a
central limit theorem for $K_N$. The expected number of bits required
to code a string of length $N$ is thus
\begin{equation*}
  \ell\E K_N \sim \frac{\ell N}{\E \dl} 
= \frac{\ceil{\lg M}}{\E \dl} N.
\end{equation*}
For simplicity, we consider the ratio
$\kk\=\xfrac{{\lg M}}{\E \dl}$, and call it the \emph{compression rate}.
(One objective of the code is to make this ratio small.)

In Khodak's construction of such a dictionary, 
we fix a threshold $r\in(0,1)$ and construct a
parsing tree as the subtree of the complete infinite binary tree
such that the internal nodes are the strings $\aaa=\ga_1\dotsm\ga_k$
with $P(\aaa)\ge r$; the external nodes are thus the strings $\aaa$
such that $P(\aaa)<r$ but the parent, $\aaa'$ say, has $P(\aaa')\ge
r$. The phrases in the Khodak code are the external nodes in this
tree. 
For convenience, we let $R=1/r>1$.
Let $M=M(R)$ be the number of phrases in the Khodak code.

In Tunstall's construction, we are instead given a number $M$. We start with
the empty phrase and then iteratively $M-1$ times replace a phrase
$\aaa$ having maximal $P(\aaa)$ by its two children $\aaa0$ and $\aaa1$.

It is easily seen that Khodak's construction with some $r>0$
gives the same result as 
Tunstall's with $M=M(R)$. 
Conversely, a Tunstall code is almost a Khodak code, with $r$ chosen
as the smallest $P(\aaa)$ for a proper prefix $\aaa$ of a phrase; the
difference is that Tunstall's construction handles ties more flexibly;
there may be some phrases too with $P(\aaa)=r$. Thus, Tunstall's
construction may give any desired number $M$ of phrases, while
Khodak's does not. We will see that in the non-arithmetic case, this
difference is asymptotically negligible, while it is important in the
arithmetic case. (This is very obvious if $p=q=1/2$, when Khodak's code
always gives a dictionary size $M$ that is a power of 2.)

Let us first consider the number of phrases, $M=M(R)$, in Khodak's
construction with a threshold $r=1/R$. 
This is a purely deterministic problem, but we may nevertheless apply our
probabilistic renewal theory arguments. In fact, $M$, the number of
leafs in the parsing tree, equals $1$ + the number of internal nodes. Thus,
$M=1+\sum_\aaa f(RP(\aaa))$
with $f(x)\=\ett{x\ge1}$, and we may apply \refT{Tsum}. 

\begin{theorem}\label{TMKhodak}
Consider the Khodak code with threshold $r=1/R$.
  \begin{thmenumerate}
\item
  If\/ $\ln p/\ln q$ is irrational, then, as $R\to\infty$,
  \begin{equation*}
	\frac{M(R)}{R} \to \frac{1} H.	
  \end{equation*}

\item
  If\/ $\ln p/\ln q$ is rational, 
then, as $R\to\infty$,
  \begin{equation*}
	\frac{M(R)}{R}  =  \frac1H
	\cdot\frac{d}{1-e^{-d}}e^{-d\frakt{(\ln R)/d}} + o(1).
  \end{equation*}
  \end{thmenumerate}
\end{theorem}

\begin{proof}
  The non-arithmetic case follows directly from \refT{Tsum}(i), since
  $\intoo f(x)x\qww\dd x=\int_1^\infty x\qww\dd x=1$.

In the arithmetic case, we use \eqref{tsumsum}. Since
$g(t)=e^{t}\ett{t\le0}$, the sum in \eqref{tsumsum} is a geometric
series that can be summed directly:
\begin{equation*}
  \psi(t)=d\sum_{kd\le t} e^{kd-t}
=\frac{d}{1-e^{-d}} e^{d\floor{t/d}-t}
=\frac{d}{1-e^{-d}} e^{-d\frakt{t/d}}.
\qedhere
\end{equation*}
\end{proof}

\begin{remark}
  \label{R15}
In the arithmetic case (ii), $\ln P(\ga)$ is a multiple of $d$ for
any string $\ga$. Hence $M(R)$ jumps only when
$R\in\set{e^{kd}:k\ge0}$, and it suffices to consider such $R$. For
these $R$, the result can be written
\begin{equation}
  \label{p15}
M(R)\sim\frac1H\frac{d}{1-e^{-d}}R,
\qquad \ln R\in d\bbZ.
\end{equation}
\end{remark}

Next, consider the length $\dl$ of a random phrase.
We will use the notation $\dtm$ for a Tunstall code with $M$ phrases
and $\dkr$ for a Khodak code with threshold $r=1/R$.

Consider first the Khodak code. By construction, given a random string
$\Xi=\xi_1\xi_2\dotsm$, the 
first phrase in it is $\xi_1\dotsm\xi_n$ where $n$ is the smallest
integer such that $P(\xi_1\dotsm\xi_n)=e^{-S_n}<r=e^{-\ln R}$.
Hence, by \eqref{nu},
\begin{equation}\label{dkr}
  \dkr=\nu(\ln R).
\end{equation}
Hence, Theorems \ref{TA1}--\ref{TA2} immediately yield the following
(as well as convergence of higher moments).

\begin{theorem}
  \label{TKh}
For the Khodak code,
the following holds as $R\to\infty$,
with $\gss=H_2-H^2=pq\ln^2(p/q)$:
\begin{align}
  \frac{\dkr}{\ln R}& \asto \frac1H, \label{tkha}
\\
\dkr&\sim\AsN\Bigpar{\frac{\ln R}H,\,\frac{\gss}{H^3}\ln R}, \label{tkhb}
\\
\Var{\dkr}&\sim\frac{\gss}{H^3}\ln R. \label{tkhc}
\end{align}
If\/ $\ln p/\ln q$ is irrational, then
\begin{equation}\label{tkhd}
\E{\dkr}=\frac{\ln R}H+\frac{H_2}{2H^2}+o(1). 
\end{equation}
If\/ $\ln p/\ln q$ is rational, then, with $d\=\gcdpq$
given by \eqref{d},
\begin{equation}\label{tkhe}
\E{\dkr}=\frac{\ln R}H+\frac{H_2}{2H^2}
+\frac dH\Bigpar{\frac12-\Bigfrakt{\frac{\ln R}{d}}}+o(1).
\end{equation}
\end{theorem}

In the arithmetic case, as said in \refR{R15}, it suffices to consider
thresholds such that $-\ln r=\ln R$ is a multiple of $d$; in this
case \eqref{tkhe} becomes
\begin{equation}\label{tkhex}
\E{\dkr}=\frac{\ln R}H+\frac{H_2}{2H^2}
+\frac d{2H}+o(1).
\end{equation}

We analyze the Tunstall code by comparing it to the Khodak code. Thus,
suppose that $M$ is given, and increase $R$ (decrease $r$) until we
find a Khodak code with $M(R)\ge M$ phrases. (By our definitions,
$M(R)$ is right-continuous, so a smallest such $R$ exists.)
Let $M_+\=M(R)\ge M$ and $M_-\=M(R-)<M$.
Thus, there are $M_+-1$ strings $\ga$ with $P(\ga)\ge r=R\qw$, and
$M_--1$ strings with $P(\ga)>r$; consequently there are $M_+-M_-$
strings with $P(\ga)=r$.
The strings with $P(\ga)=r$ are not parsing phrases in the Khodak code
(while all their children are), but we use some of them in the
Tunstall code to achieve exactly $M$ parsing phrases. Since each
of these strings replaces two parsing phrases in the Khodak code, 
the total number of parsing phrases decreases by 1 for each used string with
$P(\ga)=r$, and thus the Tunstall code uses 
$M(R)-M=M_+-M$ parsing phrases with $P(\ga)=r$.
The length $\dtm$ of a random phrase, realized as the first phrase in $\Xi$,
equals $\dkr$ unless $\Xi$ begins with one of the phrases $\ga$ in the
Tunstall code with $P(\ga)=r$, in which case $\dtm=\dkr-1$.
The probability of the latter event is evidently $P(\ga)=r$ for each
such $\ga$, and is thus $(M(R)-M)r$. Consequently, with $R$ as above,
\begin{equation}
  \label{dtm-dkr}
\dtm=\dkr-\gdm,
\end{equation}
where $\gdm\in\setoi$ and $\P(\gdm=1)
=(M(R)-M)/R$.
We can now find the results for $\dtm$:

\begin{theorem}
  \label{TT}
For the Tunstall code,
the following holds as $M\to\infty$,
with $\gss=H_2-H^2=pq\ln^2(p/q)$:
\begin{align}
  \frac{\dtm}{\ln M}& \asto \frac1H, \label{tta}
\\
\dtm&\sim\AsN\Bigpar{\frac{\ln M}H,\,\frac{\gss}{H^3}\ln M}, \label{ttb}
\\
\Var{\dtm}&\sim\frac{\gss}{H^3}\ln M. \label{ttc}
\end{align}
If\/ $\ln p/\ln q$ is irrational, then
\begin{equation}\label{ttd}
\E{\dtm}=\frac{\ln M}H+\frac{\ln H}H+\frac{H_2}{2H^2}+o(1).
\end{equation}
If\/ $\ln p/\ln q$ is rational, then, with $d\=\gcdpq$
given by \eqref{d},
\begin{multline}\label{tte}
\E{\dtm}=\frac{\ln M}H+\frac{\ln H}H+\frac{H_2}{2H^2}
+\frac 1H\ln{\frac{\sinh(d/2)}{d/2}}
\\
+\frac dH\ppsi\Bigpar{\Bigfrakt{\frac{\ln M+\ln(H(1-e^{-d})/d)}{d}}}
+o(1),
\end{multline}
where 
\begin{equation}\label{psiT}
\ppsix(x)\=\frac{e^{dx}-1}  {e^{d}-1} -x
. 
\end{equation}
\end{theorem}

Note that $\ppsix$ is continuous, with
$\ppsix(0)=\ppsix(1)=0$. $\ppsix$ is convex and thus $\ppsix\le0$ on \oi.
In the symmetric case $p=q=1/2$, $d=H=\ln 2$ and $\ppsix(x)=2^x-1-x$,
with a minimum $-0.086071\dots$.

\begin{proof}
Let as above $R$ be the smallest number with $M(R)\ge M$; thus
$M(R)\ge M>M(R-)$.
By \refT{TMKhodak},
$\ln R=\ln M+O(1)$, so \eqref{tta}--\eqref{ttc} follow from
  \eqref{tkha}--\eqref{tkhc} and the fact that $|\dtm-\dkr|\le1$, see
  \eqref{dtm-dkr}. 

If $\ln p/\ln q$ irrational,
\refT{TMKhodak} yields $M(R)/R\to 1/H$, and thus also $M(R-)/R\to
1/H$.
Since $M(R)\ge M>M(R-)$, also 
\begin{equation}
  \label{b6}
\frac MR\to\frac1H,
\end{equation}
and further
$M(R)/M\to1$.
Consequently,
\begin{equation*}
  \E\gdm = \frac{M(R)-M}{R}=\Bigpar{\frac{M(R)}M-1}\frac MR \to0,
\end{equation*}
and thus, by \eqref{dtm-dkr}, $\E\dtm=\E\dkr-\E\gdm=\E\dkr+o(1)$.
Since also, by \eqref{b6} again, $\ln R=\ln M+\ln H+o(1)$, \eqref{ttd}
follows from \eqref{tkhd}.

In the case when $\ln p/\ln q$ is rational, we argue similarly, but we
have to be more careful. First, necessarily $R=e^{Nd}$ for some
integer $N$, see \refR{R15}. Further, \eqref{p15} applies.
Let, for convenience, 
\begin{equation}\label{beta}
\gb\=H\,\frac{1-e^{-d}}d
=H\,\frac{\sinh(d/2)}{d/2}e^{-d/2};
\end{equation}
thus \eqref{p15} can be
written
$M(R)\sim\gb\qw R$ as $R\to\infty$.
Let
\begin{equation}\label{sprut}
  x\=\frac1d\ln(\gb M)-N+1
=\frac1d\ln\frac{\gb M}R+1.
\end{equation}
Then, by these definitions and \eqref{p15},
\begin{align}
  M&=\gb\qw e^{d(N-1+x)},\label{M}\\
M(R)&=\gb\qw R(1+o(1))=\gb\qw e^{dN+o(1)},\label{MR}\\
M(R-)&=M(Re^{-d})=\gb\qw (Re^{-d})(1+o(1))=\gb\qw e^{d(N-1)+o(1)}.
\end{align}
Since $M(R-)<M\le M(R)$, we see that $o(1)\le x\le 1+o(1)$. We define
also, using \eqref{M},
\begin{equation}\label{x00}
  x_0\=\Bigfrakt{\frac{\ln(\gb M)}{d}}
=\Bigfrakt{\frac{\ln e^{d(N-1+x)}}{d}}
=\frakt x.
\end{equation}
Typically, $0\le x<1$, and then $x_0=x$, but it may happen that $x$ is
slightly below 0 and $x_0=x+1$, or that $x$ is slightly above 1 and
then $x_0=x-1$.

By \eqref{sprut}, $\ln R=\ln(\gb M)+d(1-x)$, and thus \eqref{tkhex}
yields, using \eqref{beta},
\begin{equation*}
  \begin{split}
\E{\dkr}&
=\frac{\ln (\gb M)}H+\frac{H_2}{2H^2}+\frac d{2H}+\frac dH(1-x)+o(1)
\\&
=\frac{\ln M}H+\frac{\ln H}H+\frac1H\ln\frac{\sinh(d/2)}{d/2}
+\frac{H_2}{2H^2}+\frac dH(1-x)+o(1).
  \end{split}
\end{equation*}
Furthermore, by $R=e^{dN}$, \eqref{M}, \eqref{MR} and \eqref{beta},
\begin{equation*}
  \begin{split}
\E\gdm&=\frac{M(R)-M}R
=\gb\qw(1-e^{d(x-1)})+o(1)
\\&
=\frac dH\,\frac{1-e^{xd-d}}{1-e^{-d}}+o(1)
=\frac dH\Bigpar{1-\frac{e^{xd}-1}{e^d-1}}+o(1).	
  \end{split}
\end{equation*}
Combining these, we find by \eqref{dtm-dkr} and \eqref{psiT},
\begin{equation*}
  \begin{split}
\E\dtm
&=
\E{\dkr}-\E\gdm
\\&
=\frac{\ln M}H+\frac{\ln H}H+\frac1H\ln\frac{\sinh(d/2)}{d/2}
+\frac{H_2}{2H^2}
+\frac dH\ppsix(x)
+o(1).
  \end{split}
\end{equation*}
This is almost \eqref{tte}, except that there $\ppsix(x)$ is replaced
by $\ppsix(x_0)=\ppsix(\frakt{\ln(\gb M)/d})$, see \eqref{x00}.
However, as noted above, $x\neq x_0$ can happen only when one of $x$
and $x_0$ is $o(1)$ and the other is $1+o(1)$. Since the function
$\ppsix$ is continuous and $\ppsix(0)=\ppsix(1)$, we see that in this
case
$\ppsix(x)-\ppsix(x_0)=\pm(\ppsix(1)-\ppsix(0))+o(1)=o(1)$. Hence,
$\ppsix(x)=\ppsix(x_0)+o(1)$ in all cases, and \eqref{tte} follows.
\end{proof}

\begin{remark}
  We have chosen to derive \refT{TT} from the corresponding result
  \refT{TKh} for the Khodak code. An alternative is to note that
in the Tunstall code, we obtain the random phrase length $\dtm$
by stopping $\Xi$ at  $M_+-M$ of the $M_+-M_-$ strings
  $\ga$ with $P(\ga)=r$, and all strings with smaller $P(\ga)$. 
By symmetry, we obtain the same distribution
  of the length if we stop randomly with probability $(M_+-M)/(M_+-M_-)$
whenever $P(\ga)=e^{-S_n}=r$; equivalently, we stop when
  $e^{-S_n-X_0}<r$, where $X_0$ is a random variable, independent of
  $\Xi$, with values $0$ and $\eps$, for some very small positive
  $\eps=\eps(M)$, and $\P(X_0=\eps)=(M_+-M)/(M_+-M_-)$. Consequently,
  we have $\dtm\eqd\hnu(\ln R)$, with $R$ and $X_0$ as above, and we
  can apply Theorems \ref{TA1}--\ref{TA2} (and \refR{RA1}) directly.
\end{remark}

\begin{corollary}
  The compression rate for the Tunstall code is
  \begin{equation*}
\kk\=\frac{\lg M}{\E\dtm}
=\frac H{\ln2}\lrpar{1-\frac{\ln H+H_2/2H+\gd}{\ln M}+o\bigpar{(\ln M)\qw}}
  \end{equation*}
where $\gd=0$ when $\ln p/\ln q$ is irrational while when $\ln p/\ln q$
is rational,
\begin{equation*}
  \gd\=\ln\frac{\sinh(d/2)}{d/2} 
+d\ppsix\Bigpar{\Bigfrakt{\frac{\ln M+\ln(H(1-e^{-d})/d)}{d}}},
\end{equation*}
with $d$ given by \eqref{d} and $\ppsix$ by \eqref{psiT}.

For the Khodak code, the compression rate $\lg (M(R))/\E\dkr$
is asymptotically given by
the same formula, with $\ln M$ replaced by $\ln R$, except that the
$\ppsix$ term does not appear in $\gd$.
\end{corollary}

The reason that the $\ppsix$ term does not appear for the Khodak code
is that $\dkr=\dtx{M(R)}$, and in the arithmetic case, 
we may assume that $R=e^{Nd}$, and then
for $\dtx{M(R)}$,
the argument $x_0$ of $\psi$ is 
$\frakt{\ln(\gb M(R))/d}=\frakt{\ln(R)/d+o(1)}=\frakt{N+o(1)}$
 and thus close to 0 or 1, where $\ppsix$ vanishes.

\section{A stopped random walk}\label{SX}
\citet{DrSz} consider (motivated by the study of Tunstall and Khodak codes)
walks in a region in the first quadrant bounded by two crossing lines.
Their first result, on the number of possible paths, seems to require
a longer comment, and will not be considered here. Their second result
is about a random walk in the plane taking only unit steps north or
east, which is stopped when it exits the region; the probability of an
east step is $p$ each time. Coding steps east by 1 and north by 0,
this is the same as taking our random string $\Xi$. \citet{DrSz}
study, in our notation, the exit time 
\begin{equation*}
\dkv\=\min\set{n:n>k \text{ or } S_n>V\ln 2}  
\end{equation*}
for given numbers $K$ and $V$, with $K$ integer.
We thus have
\begin{equation}
  \label{x1}
\dkv=(K+1)\bmin\nu(V\ln 2).
\end{equation}
We have here kept the notations $K$ and $V$ from \cite{DrSz}, but for
convenience we in the sequel write $V_2\=V\ln 2$.
We assume $p\neq q$, since otherwise $\dkv=(K\bmin\floor V)+1$ is
deterministic. 

We need a little more notation. Let as usual
$\phi(x)\=(2\pi)\qqw e^{-x^2/2}$ and $\Phi(x)\=\int_{-\infty}^x
\phi(y)\dd y$ be the density and distribution functions of the
standard normal distribution. Further, let
\begin{equation}
  \label{Psi}
\Psi(x)\= \int_{-\infty}^x \Phi(y)\dd y = x\Phi(x)+\phi(x).
\end{equation}
This definition is motivated by the following lemma.

\begin{lemma}
  \label{LPsi}
If\/ $Z\sim N(0,1)$, then for every real $t$,
$\E(Z\bmax t)=\Psi(t)$ and 
$\E(Z\bmin t)=-\Psi(-t)$.
Further, $\Psi(t)-\Psi(-t)=t$.
\end{lemma}

\begin{proof}
  Since $\E Z=0$,
  \begin{equation*}
	\begin{split}
\E(Z\bmax t)	  
&=
\E(Z\bmax t-Z)
=\intoo\P(Z\bmax t -Z>x)\dd x
=\intoo\Phi(t-x)\dd x
\\&
=\Psi(t).
	\end{split}
  \end{equation*}
Further, since $-Z\eqd Z$,
\begin{equation*}
  -\E(Z\bmin t) = \E\bigpar{(-Z)\bmax(-t)}
= \E\bigpar{Z\bmax(-t)}
=\Psi(-t).
\end{equation*}
Finally, $\Psi(t)-\Psi(-t)=\E\bigpar{(Z\bmax t)+(Z\bmin t)} =
\E(Z+t)=t$.
(This also follows from \eqref{Psi} and $\Phi(t)+\Phi(-t)=1$,
$\phi(-t)=\phi(t)$.) 
\end{proof}

We can now state our version of the result by \citet{DrSz}. We do not
obtain as sharp error estimates as they do
(although our bounds easily can be improved when $|K-V_2/H|$ is large
enough). 
On the other hand, our
result is more general and includes the transition region when
$V_2/H\approx K$ and both stopping conditions are important.

\begin{theorem}
  \label{TX}
Suppose that $p\neq q$ and that $V,K\to\infty$. Let\/ 
$V_2\=V\ln 2$ and\/
$\hgss\=(H_2-H^2)/H^3>0$.
\begin{thmenumerate}
  \item\label{txd1}
If\/ $(K-V_2/H)/\sqrt{V_2} \to+\infty$, then $\dkv$ is asymptotically
normal:
\begin{equation}
  \label{tx1}
\dkv\sim\AsN\Bigpar{\frac{V_2}{H},\hgss V_2}.
\end{equation}
Further, $\Var(\dkv)\sim\hgss V_2$.  

\item\label{txd2}
If\/ $(K-V_2/H)/\sqrt{V_2} \to-\infty$, then $\dkv$ is asymptotically
degenerate: 
\begin{equation}  \label{tx2}
  \P(\dkv=K+1)\to1.
\end{equation}
Further, $\Var D=o(V_2)$.
  \item\label{txd3}
If\/ $(K-V_2/H)/\sqrtvv \to a\in(-\infty,+\infty$), then $\dkv$ is
asymptotically truncated normal:
\begin{equation}
  \label{tx3}
V_2\qqw(\dkv-V_2/H)\dto (\hgs Z)\bmin a = \hgs\bigpar{Z\bmin(a/\hgs)}.
\end{equation}
with $Z\sim N(0,1)$.
Further,$$\Var(\dkv)\sim V_2\Var(\hgs Z\bmin a)
=V_2\hgss\Var(Z\bmin(a/\hgs)).$$
\item\label{txe}
In every case,
\begin{align}
  \E\dkv 
&=
\frac{V_2}H -\hgs\sqrt{V_2}\Psi\Bigparfrac{V_2/H-K}{\hgs\sqrt{V_2}} +
o(\sqrt{V_2}) \label{txe1} 
\\
&=
K-\hgs\sqrt{V_2}\Psi\Bigparfrac{K-V_2/H}{\hgs\sqrt{V_2}} +
o(\sqrt{V_2}) \label{txe2}. 
\end{align}
\item\label{txe+}
If $(K-V_2/H)/\sqrt{V_2} \ge\ln V_2$, then
\begin{equation}
  \E \dkv = \frac{V_2}H+\frac{H_2}{2H^2}+\ppsi(V_2)+o(1),
\end{equation}
where $\ppsix=0$ in the non-arithmetic case and 
$\ppsix(t)=\frac dH\bigpar{1/2-\set{t/d}}$ in the $d$-arithmetic case.
\item\label{txe-}
If $(K-V_2/H)/\sqrt{V_2} \le-\ln V_2$, then
\begin{equation}
  \E \dkv = K+1+o(1).
\end{equation}
\end{thmenumerate}
\end{theorem}

\begin{proof}
  Let 
  \begin{align*}
\tD&\=\frac{\dkv-V_2/H}{\sqrtvv}, &
\tnu&\=\frac{\nu(V_2)-V_2/H}{\sqrtvv}, \\
\tK&\=\frac{K-V_2/H}{\sqrtvv}, &
\tKi&\=\frac{K+1-V_2/H}{\sqrtvv} =\tK+o(1).
  \end{align*}
Thus, by \eqref{x1}, $\tD=\tnu\bmin\tKi$.
By \refT{TA2},
\begin{equation}
  \label{x4}
\tnu=\frac{\nu(V_2)-V_2/H}{\sqrtvv} \dto N(0,\hgss).
\end{equation}
The results on convergence in distribution in \ref{txd1}--\ref{txd3}
follow immediately, noting that in \ref{txd1}, \whp\ $\nu(V_2)<K+1$
and thus $\dkv=\nu(V_2)$; in \ref{txd3} we use \eg{} the continuous
mapping theorem on $\bmin:\bbR^2\to\bbR$.

For \ref{txe}, note first that the two expressions in \eqref{txe1} and
\eqref{txe2} are the same by \refL{LPsi}.
We may by considering subsequences assume that one of the cases
\ref{txd1}--\ref{txd3} occurs.

Next, \eqref{ta2e2} 
can be written $\E(\tnu^2)\to\hgss$, which
together with \eqref{x4} implies that $\tnu^2$ is uniformly
integrable. 
(See  \eg{} \cite[Theorem 5.5.9]{Gut}.)
In case \ref{txd3}, when $\tKi$ converges, this implies that
$\tD^2=(\tnu\bmin\tKi)^2$ also is uniformly integrable, and thus the
convergence in distribution already proved for \ref{txd3} implies
$
\E\tD\to\E (\hgs(Z\bmin(a/\hgs))) = -\hgs\Psi(-a/\hgs)$,
which yields \eqref{txe1} when $K\to a\in\bbR$; 
further,
the uniform square integrability of $\tD^2$ implies 
$\Var\tD\to\Var(\hgs Z\bmin a)$ as asserted in \ref{txd3}.

If instead $\tKi\to+\infty$, case \ref{txd1}, we may assume $\tKi>0$;
then $\tD^2=(\tnu\bmin\tKi)^2\le\tnu^2$ and thus $\tD^2$ is uniformly
integrable in this case too. Hence \eqref{tx1} implies both
$\Var(\tD)\sim\hgss$, or equivalently
$\Var\dkv\sim\hgss V_2$ as asserted in \ref{txd1}, and $\E\tD\to0$, which
yields \eqref{txe1} in this case because $\Psi(-\tK)\to0$.

Finally, if $\tKi\to-\infty$, case \ref{txd2}, we may assume that
$\tKi<0$; then
$\tKi-\tD=(\tKi-\tnu)_+\le|\tnu|$ is uniformly square integrable, and
$\tKi-\tD\pto0$ by \eqref{tx2}.
Hence $\tKi-\E\tD=\E(\tKi-\tD)\to0$, and thus \eqref{txe2} holds,
since $\Psi(\tK)\to0$ and $1=o(\sqrt{V_2})$.
Further, $\Var\tD=\Var(\tKi-\tD)\to0$, which yields $\Var D=o(V_2)$.

This completes the proof of \ref{txe}.

For \ref{txe+}, we have $\dkv\le\nu(V_2)$ and thus, by the \CSineq{}
and \refT{TA1},
\begin{equation}
  \label{x7}
  \begin{split}
\E|\dkv-\nu(V_2)| 
&\le\E\bigpar{\nu(V_2)\ett{\dkv\neq\nu(V_2)}}
\\&
\le\bigpar{\E\nu(V_2)^2}\qq \P\bigpar{\dkv\neq\nu(V_2)}\qq
\\&
=O(V_2) \P\bigpar{\dkv\neq\nu(V_2)}\qq.	
  \end{split}
\end{equation}
For $\tK\ge\ln V_2$, Chernoff's bound
\cite[Theorem 2.1]{JLR} implies, because $S_{K+1}$ is a linear
transformation of a binomial $\Bi(K+1,p)$ random variable,
\begin{equation*}
  \begin{split}
  \P(\dkv\neq\nu(V_2))
&=\P(\nu(V_2)>K+1)
=\P(S_{K+1}\le V_2)
\\&
=\P(S_{K+1}-\E S_{K+1}\le -H\tKi\sqrtvv)
\\&
\le\exp\Bigpar{-c_1{\frac{\tKi^2V_2}{K+1+ \tKi\sqrtvv}}}
\\&
\le\exp(-c_2\ln^2(V_2)).
  \end{split}
\end{equation*}
for some $c_1,c_2>0$ (depending on $p$); the last inequality is
perhaps most easily seen by considering the case $K+1\le 2 V_2/H$
(when $K+1\asymp V_2$) and $K+1> 2 V_2/H$ (when $\tKi\asymp
K/\sqrtvv$) separately.
Hence, the \rhs{} of \eqref{x7} tends to 0, 
and thus $\E\dkv=\E\nu(V_2)+o(1)$. Consequently, \ref{txe+}
follows from the formulas \eqref{ta1+a0}  and \eqref{ta1+b0} 
for $\E\nu(V_2)$ provided by
\refT{TA1+}.

The argument for \ref{txe-} is very similar.  The Chernoff bound for
$S_K$ implies  
\begin{equation*}
  \begin{split}
  \P(\dkv\neq K+1)
=\P(\nu(V_2)<K+1)=\P(S_K>V_2)
\le\exp(-c_3\ln^2(V_2)),
  \end{split}
\end{equation*}
and the \CSineq{} then implies $\E\abs{K+1-\dkv}=o(1)$, proving \ref{txe-}.
\end{proof}

\appendix

\section{Some renewal theory}
\label{App}

For the readers' (and our own) convenience, we collect here a few
standard results 
from renewal theory, sometimes in less standard versions.
See \eg\ \citet{Asmussen}, \citet{FellerII} or \citet{GutSRW} for
further details.

We suppose that $X_1,X_2,\dots$ is an \iid{} sequence of non-negative
random variables with finite  mean $\mu\=\E X>0$, and that 
$S_n\=\sumin X_i$. Moreover, we suppose that $X_0$ is independent of
$X_1,X_2,\dots$ (but $X_0$ may have a different distribution, and is
not necessarily positive) and define $\hS_n\=\sumn X_i=S_n+X_0$.
We further define the first passage times $\nu(t)$ and $\hnu(t)$ by
\eqref{nu} and \eqref{hnu} and the renewal function $U$ by \eqref{U}.
(Recall that $\nu$ is a special case of $\hnu$ with $X_0=0$. Hence the
results stated below for $\tnu$ hold for $\nu$ too.)

For some theorems, 
we have to distinguish between the arithmetic (lattice) and
non-arithmetic (non-lattice) cases,
in general defined as follows:
\begin{description}
  \item[arithmetic (lattice)]
There is a positive real number $d$ such that $X_1/d$ always is an integer.
We let $d$ be the largest such number and say that $X_1$ is
\emph{$d$-arithmetic}. (This maximal $d$ is called the \emph{span} of
the distribution.) 

  \item[non-arithmetic (non-lattice)]  
No such $d$ exists. (Then $X_1$ is not supported
  on any proper closed subgroup of $\bbR$.)
\end{description}

\begin{theorem}
  \label{TA1}
As $t\to\infty$,
\begin{equation}\label{ta1a}
  \frac{\hnu(t)}{t}\asto \frac{1}{\mu}.
\end{equation}
If further $0<r<\infty$ and $\E|X_0|^r<\infty$, then
$\hnu(t)/t\to\mu\qw$ in $L^r$, \ie, $\E|\hnu(t)/t-\mu\qw|^r\to0$, and thus
\begin{equation}\label{ta1b}
 \E\lrpar{ \frac{\hnu(t)}{t}}^r\to \frac{1}{\mu^r}.
\end{equation}
\end{theorem}

\begin{proof}
See \eg{} \citetq{Theorem 2.5.1}{GutSRW} for
  the case $X_0=0$; the general case follows by essentially the same proof.
\end{proof}

\begin{theorem}
  \label{TA1+}
Suppose that $\E X_1^2<\infty$ and $\E|X_0|<\infty$.
  \begin{romenumerate}
\item
If the distribution of $X_1$ is non-arithmetic,
  then, as \ttoo,
\begin{align}\label{ta1+a0}
  \E{\nu(t)}&= \frac{t}{\mu}+\frac{\E X_1^2}{2\mu^2}+o(1)
\intertext{and, more generally,}
\label{ta1+a}
  \E{\hnu(t)}&= \frac{t}{\mu}+\frac{\E X_1^2}{2\mu^2} -\frac{\E X_0}{\mu}+o(1).
\end{align}
\item
If the distribution of $X_1$ is $d$-arithmetic,	
  then, as \ttoo,
\begin{align}\label{ta1+b0}
  \E{\nu(t)}&= 
\frac{t}{\mu}+\frac{\E X_1^2}{2\mu^2}
+\frac{d}{\mu}\Bigpar{\frac12-\Bigfrakt{\frac{t}d}}
  +o(1).
\intertext{and, more generally,}
\label{ta1+b}
  \E{\hnu(t)}&= 
\frac{t}{\mu}+\frac{\E X_1^2}{2\mu^2}
+\frac{d}{\mu}\Bigpar{\frac12-\E\Bigfrakt{\frac{t-X_0}d}}
  -\frac{\E X_0}{\mu} +o(1).
\end{align}
  \end{romenumerate}
\end{theorem}

\begin{proof}
See \eg{} \citetq{Theorem 2.5.2}{GutSRW} for
  the case $X_0=0$; the general case follows easily by conditioning on
  $X_0$. In the arithmetic case, note that
$\hnu(t)=\nu(t-X_0)=\nu\bigpar{\floor{(t-X_0)/d}d}$ and use 
$\E(\floor{(t-X_0)/d}d) = t-\E X_0 -d\E\frakt{(t-X_0)/d}$.
\end{proof}

\begin{theorem}
  \label{TA2}
Assume that $\gss\=\Var X_1<\infty$.
Then, as \ttoo, 
\begin{equation}\label{ta2}
  \frac{\hnu(t)-t/\mu}{\sqrt{t}} \dto N\Bigpar{0,\frac{\gss}{\mu^3}}.
\end{equation}
If further $\gss>0$, this can be written
$\hnu\sim\AsN(\mu\qw t,\gss\mu\qwww t)$.

Moreover, if also $\E X_0^2<\infty$, then 
\begin{align}
\label{ta2var}
  \Var\bigpar{\hnu(t)}=\frac{\gss}{\mu^3} t+o(t);
\intertext{and}
\label{ta2e2}
  \E\bigpar{\hnu(t)-t/\mu}^2=\frac{\gss}{\mu^3} t+o(t).  
\end{align}
\end{theorem}
\begin{proof}
See \eg{} \citetq{Theorem 2.5.2}{GutSRW}
for the case $X_0=0$,
noting that  \eqref{ta2var} and \eqref{ta2e2} are equivalent because
$\E\hnu(t)-t/\mu=O(1)$ by \refT{TA1+};
again, the case with a general
  $X_0$ is similar, or follows by conditioning on $X_0$.
The case $\gss=0$ is trivial.
\end{proof}

\begin{remark}\label{RA1}
  We can allow $X_0=X_0\sssn$ to depend on $n$ in Theorems
\ref{TA1}--\ref{TA2} provided $\asto$ is  weakened to $\pto$ in \eqref{ta1a}
and we add the following uniformity assumptions:
$X_0\sssn$ is tight;
for $L^r$ convergence and \eqref{ta1b} we further
assume that $\sup_n \E|X_0\sssn|^r<\infty$; for \refT{TA1+} we
assume that $X_0\sssn$ are uniformly integrable; for 
\eqref{ta2var} and \eqref{ta2e2} we
assume that $\sup_n \E|X_0\sssn|^2<\infty$.
\end{remark}

For the evaluation of \eqref{ta1+b} when $X_0$ is non-trivial, we note 
the following formula.
\begin{lemma}
  \label{Lsofie}
Suppose that $X$ has a continuous distribution with finite mean, and a
characteristic function $\gf(t)\=E e^{\ii t X}$ that satisfies
$\gf(t)=O(|t|^{-\gd})$ for some $\gd>0$. Then, for any real $u$, 
\begin{equation*}
  \E\frakt{X+u}
=
\frac12-\sum_{n\neq0}\frac{\gf(2\pi n)}{2\pi n \ii} e^{2\pi\ii n u}.
\end{equation*}
\end{lemma}

\begin{proof}
  Let $X_u\=\floor{X+u}-u+1$. Then $\frakt{X+u}=X-X_u+1$, and the
  result follows from the formula for $\E X_u$ in \cite[Theorem 2.3]{SJ175}.
\end{proof}

For the next theorem  
(known as the \emph{key renewal theorem}),
we say that a function $f\ge0$
on $\oooo$ is \emph{directly Riemann integrable} if the upper and lower 
Riemann sums $\sum_{k=-\infty}^\infty h\sup_{[(k-1)h,kh)}f $
and $\sum_{k=-\infty}^\infty h\inf_{[(k-1)h,kh)}f$
are finite and converge to the same limit as $h\to0$.
(See further \citetq{Section XI.1}{FellerII}; Feller considers
  functions on $\ooo$, but this makes no difference.)
For most purposes, the following sufficient condition suffices.
(Usually, one can take $F=f$.) 

\begin{lemma}\label{LRiemann}
  Suppose that $f$ is a non-negative function on $\oooo$.
If $f$ is bounded and \aex{} continuous, and there exists an
integrable function $F$ with $0\le f\le F$ such that
$F$ is non-decreasing on $(-\infty,-A)$ and non-increasing on
$(A,\infty)$ for some $A$, then $f$ is directly Riemann integrable.
\end{lemma}

\begin{proof}[Sketch of proof]
  It is well-known that the boundedness and \aex{} continuity implies
  Riemann integrability on any finite interval $[-B,B]$. Using the
  dominating function $F$, one sees that the tails of the Riemann sums
  coming from intervals $[(k-1)h,kh)$ outside $[-B,B]$ can be made
  arbitrarily small, uniformly in $h\in(0,1]$, by choosing $B$ large.
\end{proof}

\begin{theorem}\label{TAkey}
  Let $f$ be any non-negative directly Riemann integrable function on
  $(-\infty,\infty)$. 
  \begin{romenumerate}
\item
If the distribution of $X_1$ is non-arithmetic,
  then, as \ttoo,
  \begin{align}
\intoo f(s-t)\dd U(s) &\to \frac1\mu\intoooo f(s)\dd s,
\label{key+}
\\
\intoo f(t-s)\dd U(s) &\to \frac1\mu\intoooo f(s)\dd s.
\label{key-}
  \end{align}
\item
If the distribution of $X_1$ is $d$-arithmetic,	
  then, as \ttoo,
  \begin{align}
\intoo f(s-t)\dd U(s)&= \frac1\mu\psi(t) +o(1),
\label{keya+}
\\
\intoo f(t-s)\dd U(s) &= \frac1\mu\psi(-t) +o(1),
\label{keya-}
  \end{align}
where $\psi(t)$ is the bounded $d$-periodic function 
\begin{equation}
  \label{psi}
\psi(t)\=d\sumkoooo f(kd-t);
\end{equation}
$\psi$ has the Fourier series
\begin{equation}\label{fou}
  \psi(t)\sim\summoooo\hpsi(m)e^{2\pi\ii mt/d}
\end{equation}
with
\begin{equation}\label{fouf}
  \hpsi(m)=\widehat f(-2\pi m/d) = \intoooo e^{2\pi\ii m t/d}f(t)\dd t.
\end{equation}
  \end{romenumerate}
\end{theorem}

In particular, the average of $\psi$ is $\hpsi(0)=\intoooo f$.
The series \eqref{psi} converges uniformly on $[0,d]$; thus
$\psi$ is continuous if $f$ is. Further, if $f$ is sufficiently smooth
(an integrable second derivative is enough), then the Fourier series
\eqref{fou} converges uniformly.

\begin{proof}
  The two formulas \eqref{key+} and \eqref{key-} are equivalent by the
  substitution $f(x)\to f(-x)$.
The theorem is usually stated in the form \eqref{key-} for functions 
$f$ supported on $\ooo$; then the integral is $\int_0^t f(t-s)\dd U(s)$.
However,  the proof in \citetq{Section XI.1}{FellerII} applies to the
more general  form above as well. (The proof is based
  on approximations with step functions and the special case when
  $f(x)$ is an indicator fuction of an interval; the latter case is
  known as \emph{Blackwell's renewal theorem}.)
In fact, a substantially more general version of \eqref{key-}, where also the
  increments $X_k$ may take negative values, is given 
in  \cite[Theorem 4.2]{AN78}.

Part (ii) follows similarly (and more easily) from the fact that the
measure $\ddx U$ is concentrated on \set{kd:k\ge0}, and thus 
$\intoo f(s-t)\dd U(s)- \frac1\mu\psi(t) =\sumkoooo f(kd-t)(\ddx
U\set{kd}-d/\mu)$ together with the renewal theorem $\ddx U\set{kd}-d/\mu\to0$
as $k\to\infty$. The Fourier coefficient calculation in \eqref{fouf}
is straightforward and standard.
\end{proof}

Finally, we consider a situation where we are given also another
sequence $Y_1,Y_2,\dots$ of random variables such that the pairs
$(X_i,Y_i)$, $i\ge1$, are \iid, while $Y_i$ and $X_i$ may be (and
typically are) dependent on each other. ($Y_i$ need not be positive.)
We denote the means by $\mu_X\=\E X_1$ and $\mu_Y\=\E Y_1$; thus
$\mu_X=\mu$ in the earlier notation, and we assume as above that
$0<\mu_X<\infty$. We also suppose that $X_0$ is independent of all
$(X_i,Y_i)$, $i\ge1$.
Let $V_n\=\sumin Y_i$.

\begin{theorem}
  \label{TAVtaut}
Suppose that $\gss_X\=\Var X_1<\infty$ and $\gss_Y\=\Var Y_1<\infty$,
and let
\begin{equation*}
  \tgs^2\=\Var(\mu_X Y_1 - \mu_Y X_1).
\end{equation*}
Then 
\begin{equation*}
  \frac{V_{\hnu(t)}-\frac{\mu_Y}{\mu_X}t}{\sqrt t}
\dto N\Bigpar{0,\frac{\tgs^2}{\mu_X^3}}.
\end{equation*}
If\/ $\tgs^2>0$, this can also be written as
\begin{equation*}
V_{\hnu(t)}\sim
\AsN\lrpar{\frac{\mu_Y}{\mu_X}t,\,
\frac{\tgs^2}{\mu_X^3}t}.
\end{equation*}
\end{theorem}

Note that the special case $Y_i=1$ yields \eqref{ta2}.

\begin{proof}
  For $X_0=0$, and thus $\hnu(t)=\nu(t)$, this is \citetq{Theorem
  4.2.3}{GutSRW}. The general case follows by the same proof, or by
  conditioning on $X_0$. 
\end{proof}

\begin{remark}
  \label{RVtaut}
Again, we can allow $X_0=X_0\sss n$ to depend on $n$, as long as the
$X_0\sss n$ is tight. 
\end{remark}

\newcommand\AAP{\emph{Adv. Appl. Probab.} }
\newcommand\JAP{\emph{J. Appl. Probab.} }
\newcommand\JAMS{\emph{J. \AMS} }
\newcommand\MAMS{\emph{Memoirs \AMS} }
\newcommand\PAMS{\emph{Proc. \AMS} }
\newcommand\TAMS{\emph{Trans. \AMS} }
\newcommand\AnnMS{\emph{Ann. Math. Statist.} }
\newcommand\AnnPr{\emph{Ann. Probab.} }
\newcommand\CPC{\emph{Combin. Probab. Comput.} }
\newcommand\JMAA{\emph{J. Math. Anal. Appl.} }
\newcommand\RSA{\emph{Random Struct. Alg.} }
\newcommand\ZW{\emph{Z. Wahrsch. Verw. Gebiete} }
\newcommand\DMTCS{\jour{Discrete Math. Theor. Comput. Sci.} }
\newcommand\DMTCSproc{{Discrete Math. Theor. Comput. Sci. Proc.} }

\newcommand\AMS{Amer. Math. Soc.}
\newcommand\Springer{Springer-Verlag}
\newcommand\Wiley{Wiley}

\newcommand\vol{\textbf}
\newcommand\jour{\emph}
\newcommand\book{\emph}
\newcommand\inbook{\emph}
\def\no#1#2,{\unskip#2, no. #1,} 
\newcommand\toappear{\unskip, to appear}

\newcommand\webcite[1]{
\texttt{\def~{{\tiny$\sim$}}#1}\hfill\hfill}
\newcommand\webcitesvante{\webcite{http://www.math.uu.se/~svante/papers/}}
\newcommand\arxiv[1]{\webcite{arXiv:#1.}}

\def\nobibitem#1\par{}

\end{document}